%% file: main.tex
\documentclass[journal,transmag]{IEEEtran}
\usepackage{amsfonts}
\usepackage{graphicx}
\usepackage{mathrsfs}
\usepackage{amsmath}
\usepackage{amsthm}
\usepackage{amscd}
\usepackage{tabularx}
\usepackage{url}
\usepackage{multirow}
\usepackage{subfigure}
\usepackage{float}
\usepackage{comment}
\usepackage{array}
\usepackage{algorithm}
\usepackage{algpseudocode}
\usepackage{placeins}
\usepackage{xr} 
\usepackage[hidelinks]{hyperref}

\newcommand{\Max}[1]{\raisebox{0.5ex}{\scalebox{0.8}{$\displaystyle \max_{\oplus1}\;$}}}
\newcommand{\Lim}[1]{\raisebox{0.5ex}{\scalebox{0.8}{$\displaystyle \lim_{\oplus1}\;$}}}

\algnewcommand\algorithmicforeach{\textbf{for each}}
\algdef{S}[FOR]{ForEach}[1]{\algorithmicforeach\ #1\ \algorithmicdo}

\newtheorem{definition}{Definition}
\newtheorem{theorem}{Theorem}
\newtheorem{lemma}{Lemma}

\newtheorem{corollary}{Corollary}

\ifCLASSINFOpdf

\else

\fi

\begin{document}
%

\title{REAEDP: Entropy-Calibrated Differentially Private Data Release with Formal Guarantees and Attack-Based Evaluation}

\author{\IEEEauthorblockN{Bo Ma\IEEEauthorrefmark{1}, \IEEEauthorrefmark{2}
Jinsong Wu\IEEEauthorrefmark{3},\IEEEauthorrefmark{4}
Wei Qi Yan\IEEEauthorrefmark{1}
}
\IEEEauthorblockA{
\IEEEauthorrefmark{1}School of Mathematics and Computer Engineering,
 Auckland University of Technology, Auckland 1024 New Zealand,\\
\IEEEauthorrefmark{2} Resideo Technologies, Inc., Austin, TX 78759, USA,\\
\IEEEauthorrefmark{3} Department of Artificial Intelligence,
    Guilin University of Electronic Technology, Gui Lin, China, \\
\IEEEauthorrefmark{4} Department of Electrical Engineering,
    Universidad de Chile, Santiago, Chile, 
}

\thanks{Corresponding author:Bo Ma}
}

\markboth{}%
{Shell \MakeLowercase{\textit{et al.}}: Entropy Analysis and Evaluation Based on Differential Privacy Protection Method }

\IEEEtitleabstractindextext{%
\begin{abstract}
	Sensitive data release is vulnerable to output-side privacy threats such as membership inference, attribute inference, and record linkage. This creates a practical need for release mechanisms that provide formal privacy guarantees while preserving utility in measurable ways. We propose REAEDP, a differential privacy framework that combines entropy-calibrated histogram release, a synthetic-data release mechanism, and attack-based evaluation. On the theory side, we derive an explicit sensitivity bound for Shannon entropy, together with an extension to R\'enyi entropy, for adjacent histogram datasets, enabling calibrated differentially private release of histogram statistics. We further study a synthetic-data mechanism $\mathcal{F}$ with a privacy-test structure and show that it satisfies a formal differential privacy guarantee under the stated parameter conditions. On multiple public tabular datasets, the empirical entropy change remains below the theoretical bound in the tested regime, standard Laplace and Gaussian baselines exhibit comparable trends, and both membership-inference and linkage-style attack performance move toward random-guess behavior as the privacy parameter decreases. These results support REAEDP as a practically usable privacy-preserving release pipeline in the tested settings.
	Source code: \url{https://github.com/mabo1215/REAEDP.git}
\end{abstract}

}

\maketitle

\IEEEdisplaynontitleabstractindextext

\IEEEpeerreviewmaketitle

\section{Introduction}
\IEEEPARstart{R}{eleasing} or training on sensitive datasets---recommendation logs, trajectories, or health records---exposes users to re-identification and attribute inference: even after removing identifiers, linkage with auxiliary data or gradient-based attacks can recover who participated or what they did \cite{hallinan2016recommended,koudas2006record,huang2011adversarial}. This motivates the need for both strong privacy guarantees and objective ways to evaluate whether a given protection (e.g., adding noise) actually limits information leakage. Without a clear threat model and quantitative bounds, practitioners cannot know if their deployment is safe or how to tune parameters.

\paragraph{Practical privacy problem.} Data custodians need to release statistics or synthetic data while limiting membership inference, attribute inference, and record linkage \cite{koudas2006record}. Standard differential privacy (DP) mechanisms (Laplace, Gaussian) require known sensitivity; for histogram-based releases, the \emph{entropy} of the empirical distribution is a natural utility measure, but its sensitivity was not previously bounded in a form suitable for calibration.

\paragraph{Technical gap.} Prior work leaves three limitations relevant to this study: (i)~for histogram-based releases, explicit sensitivity control for Shannon entropy, and its extension to R\'enyi entropy, is not commonly integrated into a practical release pipeline; (ii)~entropy-based calibration, synthetic release, and formal differential privacy analysis are often studied separately rather than within a single framework; and (iii)~experimental evaluation is often limited to utility metrics, with less direct comparison against standard baselines and inference-style privacy attacks.

\paragraph{Contribution.} We propose \textbf{REAEDP}, a differential privacy framework for entropy-calibrated analysis and release. Our contributions are threefold. First, we derive an explicit \textbf{entropy sensitivity bound} for histogram statistics (Theorem~\ref{thm4}), which allows calibrated differentially private release and provides a bound on the change in released entropy under the adopted adjacency model. Second, we introduce a \textbf{synthetic-data mechanism} $\mathcal{F}$ in the Wiener-kernel setting, together with a privacy-test structure, and establish an $(\varepsilon,\delta)$-differential privacy guarantee under the stated parameter conditions (Theorem~\ref{thm5}). Third, we evaluate the framework through \textbf{baseline comparison} (Laplace vs.\ Gaussian), parameter and discretization studies, and \textbf{attack-based evaluation} on public tabular datasets, showing that the empirical entropy change remains below the theoretical bound in the tested regime and that attack performance weakens as $\varepsilon$ decreases. The main components are summarized algorithmically in Section~\ref{sec:algorithm}.

\subsection{Threat model}
We consider three output-side privacy threats that are relevant to data release: membership inference, attribute inference, and record linkage. The formal differential privacy analysis constrains output leakage in general. The explicit empirical evaluation in this paper focuses on membership inference and a linkage-style attack; attribute inference is treated as an important threat model but is not benchmarked experimentally in the present version. We assume an adversary who can observe the released output of the mechanism, including synthetic records or released statistics, and who may possess auxiliary information such as partial records or external linkage data. We do not consider adversaries that corrupt the mechanism itself or interfere with the training procedure. Under this setting, the role of differential privacy is to bound how much the output distribution can change when a single record is modified, thereby limiting what can be inferred about any one individual from the released output. The entropy sensitivity analysis in Theorem~\ref{thm4} is used to calibrate the noise level for histogram-based release under the adopted adjacency model.

\subsection{Motivation}
This paper is motivated by two related needs: \textbf{entropy-based calibration} for histogram release and a \textbf{formally analyzable synthetic-data mechanism}. Adding Laplace or Gaussian noise is standard in differential privacy, but the required noise scale depends on the sensitivity of the released quantity. For histogram-based statistics, Shannon entropy, together with its R\'enyi extension, provides a natural summary of uncertainty, so an explicit sensitivity bound is useful for calibrated release. In non-interactive settings, however, one may wish to release not only privatized counts but also synthetic records or function-valued summaries. The Wiener-kernel setting and mechanism $\mathcal{F}$ provide one way to study such releases under a formal $(\varepsilon,\delta)$-DP guarantee. The common theme across these components is the combination of sensitivity analysis, formal privacy accounting, and empirical evaluation.

\section{Related Work}

\subsection{Differential privacy for statistical release and synthetic data}
Differential privacy provides a principled way to bound the effect of any single record on the output distribution of a randomized mechanism \cite{dwork2006calibrating,dwork2014algorithmic}. Classical mechanisms such as the Laplace, Gaussian, and exponential mechanisms are well understood for numerical query release and empirical risk minimization. In non-interactive settings, however, releasing privatized datasets or synthetic records remains substantially more difficult, especially when one seeks both formal privacy guarantees and acceptable utility across heterogeneous data domains. Existing approaches often focus on privatizing selected statistics, marginals, or query answers rather than releasing a unified entropy-calibrated pipeline. Our work is positioned in this non-interactive release setting and combines histogram-based release analysis with a synthetic-data mechanism under a common differential privacy formulation.

\subsection{Entropy-based and information-theoretic perspectives}
Information-theoretic quantities such as Shannon entropy and R\'enyi entropy are widely used to describe uncertainty, leakage, and distinguishability in privacy-related analysis. These quantities are valuable for interpretation, but in a differential privacy pipeline they are useful only when their sensitivity under the adopted adjacency model is made explicit. For histogram-based release, this point is especially important because entropy is typically computed from empirical counts rather than released directly. Our contribution in this paper is therefore not to introduce entropy as a privacy concept per se, but to derive an explicit sensitivity bound for histogram entropy that enables calibrated release and theorem-aligned empirical validation in the setting studied here.

\subsection{Kernel and functional-space privacy mechanisms}
Differential privacy has also been studied for function-valued outputs, kernel methods, and infinite-dimensional representations, where one seeks to release privatized means, trajectories, or other structured summaries rather than only finite-dimensional tabular statistics. Such settings are attractive because they can encode richer geometric or temporal structure, but they also raise practical questions about approximation, implementation, and utility degradation. In this paper, the Wiener-kernel setting is used to study a synthetic-data mechanism with a tunable penalty parameter $\rho$, allowing us to analyze the privacy--utility tradeoff in a functional-space formulation while keeping the main privacy accounting in an $(\varepsilon,\delta)$-DP framework.

\subsection{Privacy attacks and empirical evaluation}
Formal privacy guarantees and empirical privacy evaluation play complementary roles. Differential privacy controls worst-case output leakage at the mechanism level, whereas attack-based evaluation provides an operational view of what can still be inferred in a tested release setting. Membership inference, attribute inference, and record linkage are standard output-side threat models in the privacy literature. In practice, many papers report utility only, or report privacy indirectly through noise magnitude or parameter choice. In contrast, this paper includes explicit membership-inference and linkage-style evaluation so that the observed privacy behavior can be examined alongside the theorem-based guarantees. Attribute inference is part of the threat model considered here, but it is not benchmarked experimentally in the present version.

\subsection{Relation to private optimization}
Differential privacy is also widely studied in private optimization and learning, especially through noisy stochastic gradient methods \cite{song2013stochastic,bassily2014private,abadi2016deep}. Those methods address privacy during model training and rely on gradient clipping, noise injection, and privacy composition. Our focus is different: we study non-interactive release of histogram-based statistics and synthetic outputs rather than private optimization of predictive models. We therefore mention noisy stochastic gradient methods only as related conceptual background, not as a direct empirical baseline.

\begin{table*}[t]
\centering
\caption{Structured comparison with prior work (Method, Data Type, Formal DP, Synthetic Release, Entropy Calibration, Attack Eval, Difference).}
\label{tab:related}
\small
\begin{tabular}{@{}p{2.4cm}|p{1.6cm}|p{2cm}|p{1.2cm}|p{1.6cm}|p{1.6cm}|p{3.6cm}@{}}
\hline
\textbf{Work} & \textbf{Data} & \textbf{Formal DP} & \textbf{Synthetic} & \textbf{Entropy cal.} & \textbf{Attack eval.} & \textbf{Difference} \\
\hline
Dwork et al.\ \cite{dwork2006differential,dwork2014algorithmic} & Queries, stats & $\varepsilon$ or $(\varepsilon,\delta)$ & No (or separate) & No & No & Foundational; no entropy sensitivity or synthetic $\mathcal{F}$. \\
Noisy-SGD \cite{song2013stochastic,bassily2014private} & Gradients, ML & $(\varepsilon,\delta)$ (composition) & No & No & Limited & DP learning; no histogram entropy or synthetic release. \\
DP histogram / Laplace \cite{dwork2014algorithmic} & Counts, histograms & $\varepsilon$-DP & No & Sensitivity for counts only & No & No entropy sensitivity; we give $\Delta_H$ and calibration. \\
Kernel / RKHS DP (prior) & Function, kernel mean & Some $(\varepsilon,\delta)$ & Varies & No & Rare & We give full mechanism $\mathcal{F}$, algorithm, and MIA. \\
\hline
\textbf{This work (REAEDP)} & Histogram, tabular, RKHS & $(\varepsilon,\delta)$ (Theorems~\ref{thm4},\ref{thm5}) & Yes ($\mathcal{F}$) & Yes ($\Delta_H$, $\Delta_{H_\alpha}$) & Yes (MIA, Fig.~\ref{fig:mia}) & Combines entropy-sensitive analysis, a synthetic-data mechanism under stated conditions, and empirical attack-based evaluation. \\
\hline
\end{tabular}
\end{table*}

\subsection{Position of this work}
Relative to the literature above, this paper contributes a unified treatment of three elements that are often considered separately: entropy-sensitive histogram analysis, a synthetic-data release mechanism with a formal $(\varepsilon,\delta)$-DP guarantee under stated conditions, and empirical privacy evaluation through baseline comparison and attack-based testing. The goal is not to claim a universal replacement for existing differential privacy methods, but to provide a concrete release framework whose theoretical and empirical components can be examined together in the tested setting.

\section{Preliminaries}
\subsection{Adjacent Datasets}
We use two adjacency notions in this paper, depending on the result under study. For \textbf{histogram entropy sensitivity} (Theorem~\ref{thm4}), we use \emph{replacement adjacency}: two datasets of the same size $n$ are adjacent if they differ in exactly one record. For the \textbf{synthetic-data mechanism} $\mathcal{F}$ and its $(\varepsilon,\delta)$-DP analysis (Theorem~\ref{thm5}), we use \emph{add/remove adjacency} of the form $D' = D \cup \{d'\}$ or $D = D' \cup \{d'\}$. The relevant adjacency model is stated explicitly in each theorem.

\subsection{Differential Privacy and Composition}
\begin{definition}[$(\varepsilon,\delta)$-Differential Privacy~\cite{dwork2006differential}]
A randomized mechanism $\mathcal{M}$ with domain $\mathcal{X}^n$ and range $\mathcal{Y}$ satisfies $(\varepsilon,\delta)$-differential privacy if for any two adjacent datasets $D$, $D'$, where adjacency is defined according to the mechanism or theorem under consideration (Section~III-A and each theorem statement), and any measurable set $S \subseteq \mathcal{Y}$,
\begin{equation}
\mathbb{P}\mathrm{r}\{\mathcal{M}(D) \in S\} \leq e^{\varepsilon}\,\mathbb{P}\mathrm{r}\{\mathcal{M}(D') \in S\} + \delta.
\end{equation}
\end{definition}

\begin{theorem}[Sequential Composition]\label{thm1}
If mechanism $\mathcal{M}_i$ is $(\varepsilon_i,\delta_i)$-DP for $i=1,\ldots,k$, then the composition $(\mathcal{M}_1(D),\ldots,\mathcal{M}_k(D))$ is $\bigl(\sum_{i=1}^k \varepsilon_i,\, \sum_{i=1}^k \delta_i\bigr)$-DP. See Appendix~\ref{app:proofs} for references.
\end{theorem}

\begin{theorem}[Advanced Composition~\cite{dwork2014algorithmic}]\label{thm2}
Under adaptive composition of $k$ mechanisms each satisfying $(\varepsilon_0,\delta_0)$-DP, the overall mechanism satisfies $\bigl(\varepsilon,\, k\delta_0 + \delta'\bigr)$-DP with $\varepsilon = \sqrt{2k\ln(1/\delta')}\,\varepsilon_0 + k\varepsilon_0(e^{\varepsilon_0}-1)$ for any $\delta'>0$. See Appendix~\ref{app:proofs} for references.
\end{theorem}

\subsection{Entropy and Sensitivity}
Let $z$ be a random variable on a finite set with histogram counts $(c_1,\ldots,c_m)$, and let $n=\sum_{i=1}^m c_i$. The Shannon entropy is
\begin{equation}
H(z) = -\sum_{i=1}^m \frac{c_i}{n} \log_2 \frac{c_i}{n} = \log_2 n - \frac{1}{n}\sum_{i=1}^m c_i\log_2 c_i.
\end{equation}
The \emph{R\'enyi entropy} of order $\alpha>0$, $\alpha\neq 1$, is
\begin{equation}
H_\alpha(z) = \frac{1}{1-\alpha}\log_2 \Bigl(\sum_{i=1}^m \Bigl(\frac{c_i}{n}\Bigr)^\alpha\Bigr).
\end{equation}
As $\alpha\to 1$, $H_\alpha(z)\to H(z)$. For adjacent datasets (in the replacement sense above), histograms differ in at most two positions; the \emph{sensitivity} of $H$ is defined as $\Delta_H = \max_{D\sim D'} |H(z)-H(z')|$.

\begin{theorem}[Shannon entropy sensitivity under replacement adjacency]\label{thm4}
\emph{Assumption:} $D$ and $D'$ are adjacent under the \emph{replacement adjacency} model of Section~III-A (same size $n$, differ in exactly one record), with histograms of $m$ bins. \emph{Claim:} The sensitivity of Shannon entropy $H$ satisfies the \emph{upper bound}
\begin{equation}
\Delta_H \leq \frac{1}{n}\Bigl(2 + \frac{1}{\ln 2} + 2\log_2 n\Bigr).
\end{equation}
The bound has the correct $1/n$ order; see Appendix~\ref{app:proofs} for the proof.
\end{theorem}

\subsection{Wiener Kernel and RKHS}
In the Wiener kernel (RKHS) setting, a kernel has the form $K(x,\cdot) = \sum_{i=1}^\infty \lambda_i \psi_i(x)\varphi_i(\cdot)$. The penalty parameter $\rho$ controls the regularization; the private release of the mean in this space is sanitized so that utility is preserved while satisfying $(\varepsilon,\delta)$-DP. Larger $\rho$ reduces effective noise and brings the private mean closer to the original mean; smaller $\rho$ increases privacy but may affect training divergence in federated settings.

\subsection{Problem formulation}
We formalize the release setting as follows. Let $D$ denote the sensitive dataset (e.g., histogram counts or raw records). A \emph{mechanism} is a randomized map $\mathcal{M}: D \mapsto \widetilde{D}$ (or $\mathcal{M}: D \mapsto \tilde{y}$ for a single output), where $\widetilde{D}$ (or $\tilde{y}$) is the released, privatized object. \textbf{Privacy goal:} $\mathcal{M}$ satisfies $(\varepsilon,\delta)$-DP under the relevant adjacency model in Section~III-A, so that output-side information leakage is formally constrained. In the experiments, we assess this behavior explicitly through membership-inference and linkage-style attacks. \textbf{Utility goal:} the released output should preserve useful statistics (e.g., entropy of histograms, or mean in an RKHS) within a controllable error. \textbf{Attack model:} the adversary observes the release and possibly auxiliary information; we evaluate privacy via membership inference attack (MIA) in Section~\ref{sec:experiments}.

\section{Method and Main Results}

\subsection{Notation and symbols}
Table~\ref{tab:notation} lists the main notation used in the paper for differential privacy, entropy, the mechanism, and the proofs.
\begin{table*}[t]
\centering
\caption{Notation and symbols used in the paper (method and proofs).}
\label{tab:notation}
\small
\begin{tabular}{@{}p{2.2cm}p{12.2cm}@{}}
\hline
\textbf{Symbol} & \textbf{Meaning} \\
\hline
$D$, $D'$ & Adjacent datasets. For Theorem~\ref{thm4}, adjacency means replacement (same size $n$, differ in one record). For Theorem~\ref{thm5}, adjacency means add/remove: $D' = D \cup \{d'\}$ or $D = D' \cup \{d'\}$. \\
$\mathcal{M}$, $\mathcal{F}$ & Randomized mechanism; $\mathcal{F}$ is our synthetic-data mechanism with range $\mathcal{U}$. \\
$\varepsilon$, $\delta$ & Privacy parameters: $(\varepsilon,\delta)$-DP means $\mathbb{P}\mathrm{r}\{\mathcal{M}(D)\in S\} \leq e^\varepsilon \mathbb{P}\mathrm{r}\{\mathcal{M}(D')\in S\} + \delta$. \\
$H(z)$, $H_\alpha(z)$ & Shannon entropy and R\'enyi entropy (order $\alpha$) of distribution $z$. \\
$\Delta_H$ & Sensitivity of $H$: maximum $|H(z)-H(z')|$ over adjacent histograms. \\
$n$, $m$ & $n$ = number of records in the dataset; $m$ = number of histogram bins. \\
$c_i$, $c'_i$ & Histogram counts (bin $i$) for $D$ and $D'$; adjacent histograms differ in at most two positions. \\
$\mathcal{U}$ & Universe of possible (synthetic) records; mechanism output $y \in \mathcal{U}$. \\
$I_s(y)$ & Partition index of record $s$ for candidate $y$; $C_j(D,y) = \{s \in D : I_s(y) = j\}$. \\
$C_j(D,y)$ & Set of seeds (records) in partition $j$ that are consistent with $D$ and $y$. \\
$\mathrm{pt}(D,j,y)$ & Privacy-test pass probability: $\mathbb{P}\mathrm{r}\{L \geq k - |C_j(D,y)|\}$; $L$ is a random variable. \\
$q(D,j,y)$ & $\mathrm{pt}(D,j,y) \sum_{s \in C_j(D,y)} p_s(y)$; used in $\mathbb{P}\mathrm{r}\{\mathcal{F}(D)=y\} = \frac{1}{|D|}\sum_{j \geq 0} q(D,j,y)$. \\
$p_s(y)$ & Probability mass assigned by seed $s$ to output $y$ (mechanism design). \\
$k$, $t$, $\gamma$ & Parameters: $k$ = privacy threshold (larger = stricter); $t$ = partition threshold; $\gamma$ = tuning parameter. \\
$\varepsilon_0$ & Base privacy parameter; $\varepsilon = \varepsilon_0 + \ln(1+\gamma/t)$; $\delta = e^{-\varepsilon_0(k-t)}$. \\
$Y_{t-}$, $Y_{t+}$ & In Theorem~\ref{thm5} proof: $Y_{t-} = \{y \in Y : c(d',y) < t\}$, $Y_{t+} = \{y \in Y : c(d',y) \geq t\}$; $c(d',y) = |C_{I_{d'}(y)}(D,y)|$. \\
$\rho$ & Penalty parameter in the Wiener kernel (RKHS) setting; larger $\rho$ improves utility. \\
$K$, $\lambda_i$, $\psi_i$, $\varphi_i$ & Wiener kernel: $K(x,\cdot) = \sum_{i=1}^\infty \lambda_i \psi_i(x) \varphi_i(\cdot)$. \\
\hline
\end{tabular}
\end{table*}

\subsection{Algorithmic implementation}\label{sec:algorithm}
The following pseudocode specifies the histogram-based DP release (entropy-calibrated) and the high-level synthetic mechanism $\mathcal{F}$. \textbf{Input:} Dataset $D$, privacy parameters $(\varepsilon,\delta)$, histogram bins $m$, optional kernel parameters ($\rho$ for Wiener). \textbf{Output:} Noisy histogram (counts), or synthetic record $y \in \mathcal{U}$, or private RKHS mean. \textbf{Complexity:} Histogram release: $O(|D| + m)$; mechanism $\mathcal{F}$: depends on partition and privacy-test implementation (see Appendix).
\begin{algorithm}[h]
\caption{REAEDP: Histogram release with entropy bound and synthetic mechanism}
\begin{algorithmic}[1]
\State \textbf{Input:} $D$, $\varepsilon$, $\delta$, $m$ (bins), $k$, $t$, $\gamma$, $\varepsilon_0$
\State Compute histogram $(c_1,\ldots,c_m)$ from $D$; $n \gets \sum_i c_i$
\State $\Delta_H \gets (1/n)(2 + 1/\ln 2 + 2\log_2 n)$ \Comment{Theorem~\ref{thm4} bound}
\State Add Laplace noise to counts: $\tilde{c}_i \gets c_i + \mathrm{Lap}(1/\varepsilon)$; clip $\tilde{c}_i \geq 0$
\State \textbf{Output (histogram):} $(\tilde{c}_1,\ldots,\tilde{c}_m)$; entropy $H(\tilde{c})$
\State \textbf{For synthetic $\mathcal{F}$:} sample seed $s \sim D$; sample candidate $y \sim p_s$; compute $C_j(D,y)$, run privacy test $\mathrm{pt}(D,j,y)$; if pass, output $y$, else resample
\end{algorithmic}
\end{algorithm}

We consider a mechanism $\mathcal{F}$ that, given a dataset $D$, outputs a synthetic record $y \in \mathcal{U}$. For each seed record $s \in D$, a candidate $y$ is sampled from a seed-dependent proposal distribution $p_s(\cdot)$. For a fixed candidate $y$, the seeds are grouped by a partition index $I_s(y)$, and $C_j(D,y)$ denotes the subset of seeds in partition $j$ associated with $y$. The privacy test accepts a candidate from partition $j$ with probability
\[
\mathrm{pt}(D,j,y)=\Pr\{L \ge k-|C_j(D,y)|\},
\]
where $L$ is an integer-valued random variable independent of the sampling step and $k$ is the privacy-test threshold. The parameters $k$, $t$, $\gamma$, and $\varepsilon_0$ determine the acceptance behavior and the resulting privacy guarantee, with
\[
\varepsilon=\varepsilon_0+\ln\!\left(1+\frac{\gamma}{t}\right).
\]
A complete theorem-level specification is provided in the appendix.

\begin{lemma}\label{lemma:2yinU}
For any fixed $y \in \mathcal{U}$ and dataset $D^*$,
\begin{equation}
\mathbb{P}\mathrm{r}\{\mathcal{F}(D^*)=y\} = \frac{1}{|D^*|}\sum_{s \in D^*} \bigl[ p_s(y)\,\mathbb{P}\mathrm{r}\{(D^*,s,y)\;\mathrm{pass\;test}\}\bigr].
\end{equation}
Proof is in Appendix~\ref{app:proofs}.
\end{lemma}

\begin{lemma}[Neighboring datasets]\label{lemma:neighboringdatasets}
Let $j = I_{d'}(y)$. Then: (a) if $i \neq j$, $\mathrm{pt}(D',i,y) = \mathrm{pt}(D,i,y)$; (b) if $i = j$, $\mathrm{pt}(D,j,y) \leq e^{\varepsilon_0}\,\mathrm{pt}(D',j,y)$. Proof is in Appendix~\ref{app:proofs}.
\end{lemma}

\begin{lemma}\label{lemma4:djpartition}
For $i \neq j$, $q(D,i,y) = q(D',i,y)$. For partition $j$ with $d' \in C_j(D',y)$: if $p_{d'}(y)>0$ then $\mathrm{pt}(D',j,y) \leq \mathrm{pt}(D,j,y)$ and $q(D,j,y) < q(D',j,y)$. Proof is in Appendix~\ref{app:proofs}.
\end{lemma}

\begin{corollary}\label{corollary:ismallthan0}
For all $i$, $q(D,i,y) \leq q(D',i,y)$ when $D' = D \cup \{d'\}$. Proof is in Appendix~\ref{app:proofs}.
\end{corollary}

\begin{lemma}\label{lemma5:random1}
Let $|D| \geq k$ and $D' = D \cup \{d'\}$. Then: (1) $\mathbb{P}\mathrm{r}\{\mathcal{F}(D)=y\} \leq e^\varepsilon\,\mathbb{P}\mathrm{r}\{\mathcal{F}(D')=y\}$; (2) $\mathbb{P}\mathrm{r}\{\mathcal{F}(D')=y\}$ is at most $e^\varepsilon\,\mathbb{P}\mathrm{r}\{\mathcal{F}(D)=y\}$ plus a term $\delta(D',d',y)$ when $|C_j(D,y)|<t$, and at most $e^\varepsilon\,\mathbb{P}\mathrm{r}\{\mathcal{F}(D)=y\}$ when $|C_j(D,y)| \geq t$. Proof is in Appendix~\ref{app:proofs}.
\end{lemma}

\begin{theorem}[Differential privacy of $\mathcal{F}$ under add/remove adjacency]\label{thm5}
	Let $D$ and $D'$ be adjacent under the add/remove adjacency model, with $D' = D \cup \{d'\}$ or $D = D' \cup \{d'\}$. Under the mechanism definition and parameter conditions stated in Appendix~\ref{app:proofs}, the mechanism $\mathcal{F}$ satisfies $(\varepsilon,\delta)$-differential privacy with
	\[
	\varepsilon = \varepsilon_0 + \ln\!\left(1+\frac{\gamma}{t}\right),
	\qquad
	\delta = e^{-\varepsilon_0(k-t)}.
	\]
\end{theorem}

\ifCLASSOPTIONcaptionsoff
  \newpage
\fi

\section{Experiments}\label{sec:experiments}

This section evaluates two main components of the paper: entropy calibration (Theorem~\ref{thm4}, $\Delta_H$ bound and calibrated histogram release) and the synthetic-data mechanism $\mathcal{F}$ (Theorem~\ref{thm5}, Wiener kernel and privacy test). Noisy-SGD details and supplementary image-noise results are provided in the supplementary appendix, with Noisy-SGD in Appendix~I and image-noise results in Appendices~J--L. The experiments are organized into four parts: privacy-test behavior (pass rate of mechanism $\mathcal{F}$), utility comparison against standard DP baselines (Laplace, Gaussian, DP synthetic), theorem-oriented validation of entropy sensitivity ($\widehat{\Delta H}$ vs.\ bound), and attack-based evaluation under membership inference and linkage-style inference. Unless otherwise stated, all experiments that invoke $\mathcal{F}$ use $(\varepsilon_0,\gamma,t,k)$ chosen from the parameter domain $\mathcal{A}$ in Appendix~\ref{app:proofs} so that the assumptions of Theorem~\ref{thm5} are satisfied.

\subsection{Experimental setup}
Unless otherwise stated, experiments were run with a fixed random seed (e.g., 42) to reduce variability across repeated trials. The Wiener kernel (RKHS) experiments use a Chi-square--like process with 50 sample paths and 80 time steps; the privacy parameter is set to $\varepsilon = 1.0$ and $\delta = 10^{-5}$ for the Gaussian mechanism. Supplementary image-noise utility results are provided in the supplementary appendix (Appendices~J--L), with detailed metric comparison in Appendix~L, as supporting evidence only. The entropy sensitivity bound is evaluated over dataset sizes $n \in \{50, 100, 200, 500, 1000, 2000, 5000\}$.

\subsection{Datasets and data sources}
Table~\ref{tab:datasets} summarizes the data used in the experiments. For the real-data CSV experiments we use multiple public tabular datasets: the \emph{Amazon Product and Google Locations Reviews}~\cite{adler2026amazon}; House Prices (Kaggle, column SalePrice); Home Credit (Kaggle, column AMT\_INCOME\_TOTAL); Tabular Feature Engineering (Kaggle, column $y_1$); synthetic regression sets (\texttt{r\_diff\_train}, \texttt{pow\_train}); and U.S. unemployment (column UNRATE). The same entropy+DP pipeline (histogram, Laplace noise, Shannon entropy) is run on each to validate the framework across domains (see Table~\ref{tab:csv_entropy_multi} and Figure~\ref{fig:csv_entropy_multi}). \textbf{Where to find more data of the same type:} The UCI Machine Learning Repository~\cite{uci2023}, Kaggle~\cite{kaggle}, and record-linkage or time-series benchmarks can be placed in \texttt{data/} and under \texttt{csv\_entropy\_experiment} or \texttt{csv\_entropy\_all} to reproduce the single- or multi-dataset histogram entropy experiments.

\begin{table*}[!h]
\centering
\caption{Datasets used in the experiments.}
\label{tab:datasets}
\begin{tabular}{llp{4.2cm}}
\hline
\textbf{Dataset} & \textbf{Type} & \textbf{Source / where to find more} \\
\hline
y\_amazon-google-large ($y$) & Tabular, numeric & UCI~\cite{adler2026amazon}; time series, $\sim$3M rows. \\
house-prices-train (SalePrice), home-credit, tabular-fe, etc. & Tabular, numeric & Kaggle~\cite{kaggle}; multi-dataset validation (Table~\ref{tab:csv_entropy_multi}). \\
Synthetic (Chi-square process) & Stochastic process & Generated (50 paths, 80 steps). \\
Synthetic $64\times 64$ image & Image & Supplementary appendix (Appendices~J--L). \\
$n \in \{50,\ldots,5000\}$ (entropy bound) & Scalar $n$ & Used only to plot Theorem~\ref{thm4} bound. \\
\hline
\end{tabular}
\end{table*}

\subsection{Privacy test pass rate (Figure~\ref{fig:generalpic})}
Figure~\ref{fig:generalpic} illustrates the empirical behavior of the privacy test used by mechanism $\mathcal{F}$. Horizontal axis: $\gamma$; vertical axis: pass rate (fraction of candidates passing the test); curves: one per $k \in \{10, 20, 30, 50\}$; $t = 2$, max check = 10,000. For suitable $\gamma$, the acceptance rate remains non-negligible even when $k$ is large, which supports the practical usability of the mechanism. The formal $(\varepsilon,\delta)$-DP guarantee is provided by Theorem~\ref{thm5} rather than by this figure alone.
\begin{figure}[t]
	\centering
	\includegraphics[width=\columnwidth]{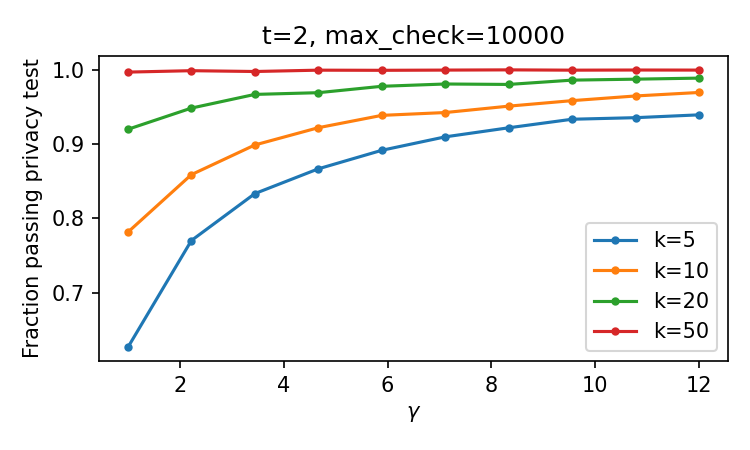}
	\caption{Privacy test pass rate vs.\ $\gamma$ for several $k$ ($t=2$).}\label{fig:generalpic}
\end{figure}

\subsection{Wiener Kernel / RKHS Private Mean}
\begin{figure*}[htbp]
	\centering
	\includegraphics[width=0.32\textwidth]{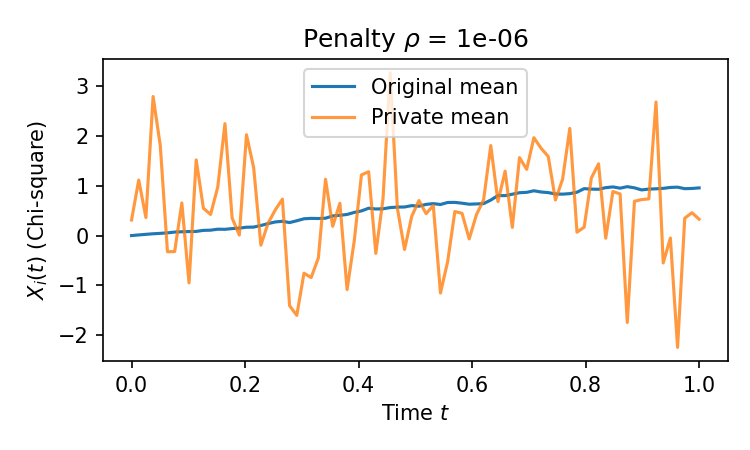}
	\hfill
	\includegraphics[width=0.32\textwidth]{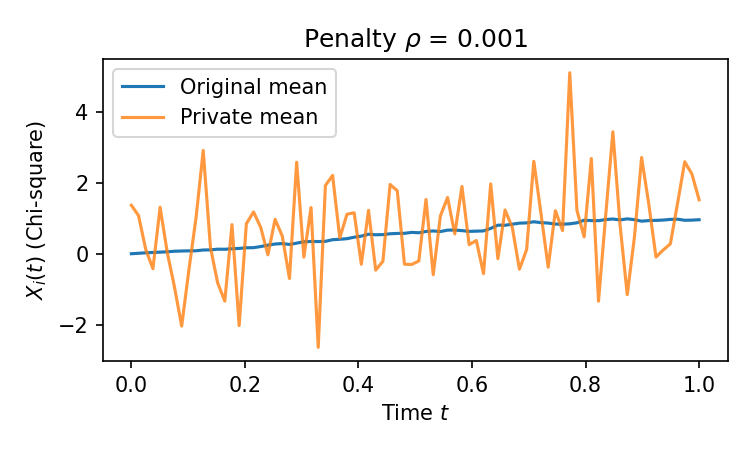}
	\hfill
	\includegraphics[width=0.32\textwidth]{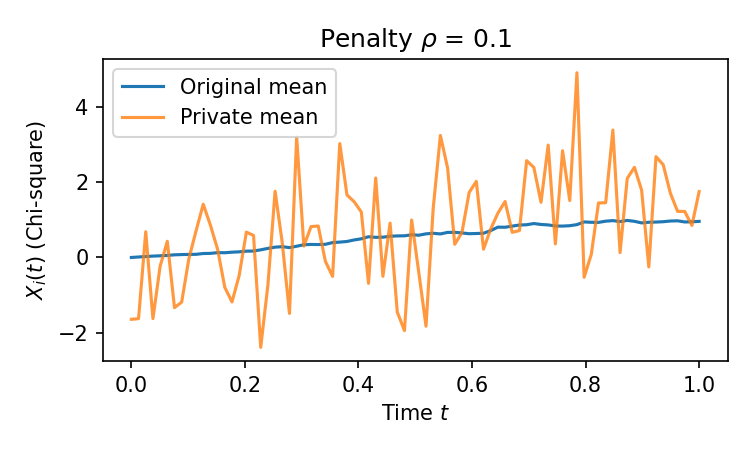}
	\caption{Wiener kernel: original vs.\ private mean for $\rho = 10^{-6}$, $0.001$, $0.1$ (left to right).}\label{fig:GaussianKernel1}
\end{figure*}

In Figure~\ref{fig:GaussianKernel1} the horizontal axis is time $t \in [0,1]$ and the vertical axis is the Chi-square--like process value $X_i(t)$. The blue curve is the original (non-private) mean; the orange curve is the private mean under the Wiener kernel mechanism. Left panel ($\rho = 10^{-6}$): the gap between the two curves is visible over the whole time range, reflecting the added noise. Middle panel ($\rho = 0.001$): the private mean tracks the original more closely. Right panel ($\rho = 0.1$): the two curves almost overlap. Thus, increasing the penalty parameter $\rho$ reduces the effective noise and improves utility, consistent with the kernel form $K(x,\cdot) = \sum_{i=1}^\infty \lambda_i \psi_i(x) \varphi_i(\cdot)$. Practitioners should choose $\rho$ according to the desired privacy--utility tradeoff: too small $\rho$ may harm utility in federated or downstream learning; too large $\rho$ may reduce privacy protection against stronger downstream inference attacks.

\subsection{Synthetic data generation and privacy test}

The extent to which we can synthesize large privacy-protected datasets depends on the proportion of synthetic candidates that pass the privacy test. We set $t = 2$ and max check = 10,000, and vary $k$ and $\gamma$. Figure~\ref{fig:generalpic} shows the pass rate as a function of $\gamma$ for several $k$. \textbf{Observation:} For each $k$, the pass rate increases with $\gamma$; for larger $k$ the curves shift downward but remain well above zero in the tested $\gamma$ range. Thus even under strict privacy ($k$ large), a suitable choice of $\gamma$ allows generating a significant volume of synthetic data, which supports the use of the mechanism $\mathcal{F}$ in non-interactive release scenarios.

\subsection{Entropy sensitivity bound (Theorem~\ref{thm4})}

\begin{figure}[htbp]
	\centering
	\includegraphics[width=\columnwidth]{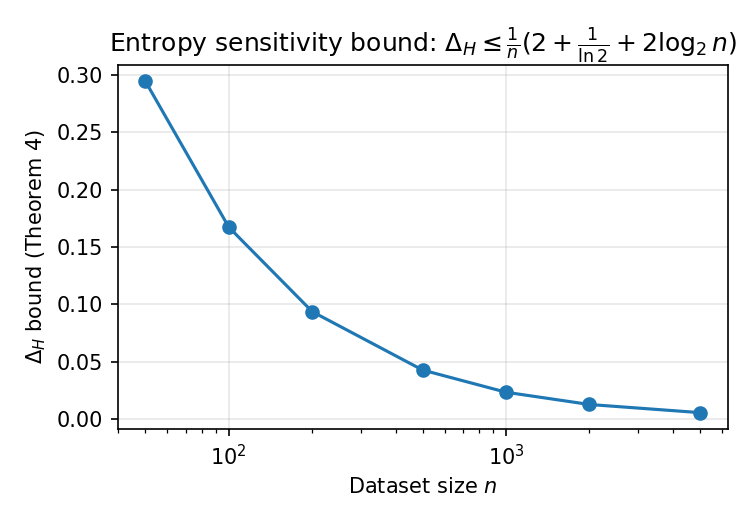}
	\caption{Entropy sensitivity bound $\Delta_H$ vs.\ dataset size $n$, illustrating the decrease of the theoretical bound used for calibrated histogram release.}\label{fig:entropy_bound}
\end{figure}

Figure~\ref{fig:entropy_bound} plots the bound $\Delta_H \leq \frac{1}{n}(2 + \frac{1}{\ln 2} + 2\log_2 n)$ for $n \in \{50, \ldots, 5000\}$ (log scale). The bound decreases as $n$ increases: for $n = 50$ it is about 0.2, and for $n = 5000$ it falls below 0.01. Together with Figure~\ref{fig:delta_h_empirical}, the empirical entropy sensitivity (mean and max over adjacent pairs) remains below the bound across the tested values of $n$, indicating that the theoretical upper bound is respected in the tested regime.

\subsection{Empirical entropy sensitivity ($\widehat{\Delta H}$ vs.\ bound)}
To align experiments with the theorem, we add a small-scale, controlled synthetic experiment: we sample or enumerate adjacent histogram pairs $(D,D')$ (replacement adjacency: one count moves from one bin to another), compute $\widehat{\Delta H} = \max_{D \sim D'} |H(D) - H(D')|$ over the sampled pairs, and compare to the theoretical bound $\Delta_H \leq \frac{1}{n}(2 + \frac{1}{\ln 2} + 2\log_2 n)$. For each $(n,m)$ we draw random histograms (multinomial with $n$ records and $m$ bins), then over adjacent pairs we compute the empirical maximum $|H(z)-H(z')|$. Figure~\ref{fig:delta_h_empirical} compares the empirical entropy sensitivity $\widehat{\Delta H}$ with the theoretical bound from Theorem~\ref{thm4}. Across the tested values of $n$, both the mean and maximum empirical sensitivity remain below the bound, indicating that the theoretical upper bound is respected in the tested regime.
\begin{figure}[htbp]
	\centering
	\includegraphics[width=\columnwidth]{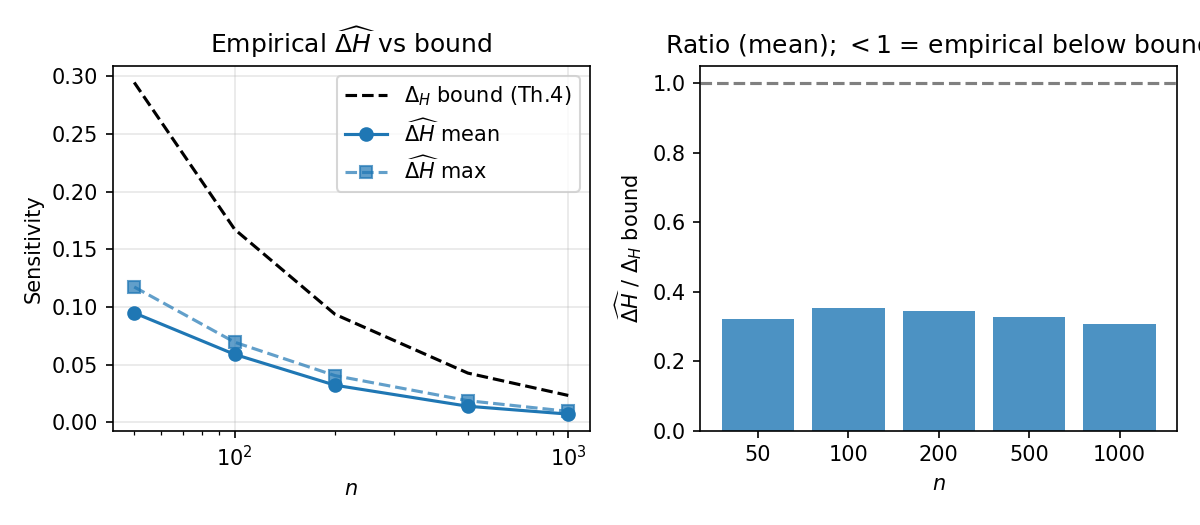}
	\caption{Empirical $\widehat{\Delta H}$ vs.\ theoretical bound (Theorem~\ref{thm4}): mean and maximum over adjacent pairs. The ratio remains below 1 in the tested regime, indicating that the empirical entropy sensitivity stays below the theoretical bound.}\label{fig:delta_h_empirical}
\end{figure}

\subsection{Baseline comparison}
Figure~\ref{fig:baseline} compares standard Laplace and Gaussian mechanisms and two \textbf{DP synthetic} baselines: (1)~Laplace-noised counts then multinomial sample to obtain synthetic histogram counts; (2)~Gaussian-noised counts then multinomial sample. Entropy preservation error $|H_{\mathrm{orig}} - H_{\mathrm{noisy}}|$ and count MAE vs.\ $\varepsilon$ are reported on a public tabular dataset (Amazon-Google reviews). For the tested values of $\varepsilon$, the observed entropy error for all four mechanisms remains below the theoretical bound $\Delta_H$ from Theorem~\ref{thm4}, and utility improves as $\varepsilon$ increases. The DP synthetic baselines illustrate that releasing synthetic counts (noisy distribution then sample) is comparable to direct noisy counts; our entropy-calibrated pipeline is comparable to these baselines and the theoretical bound holds in practice.
\begin{figure}[htbp]
	\centering
	\includegraphics[width=\columnwidth]{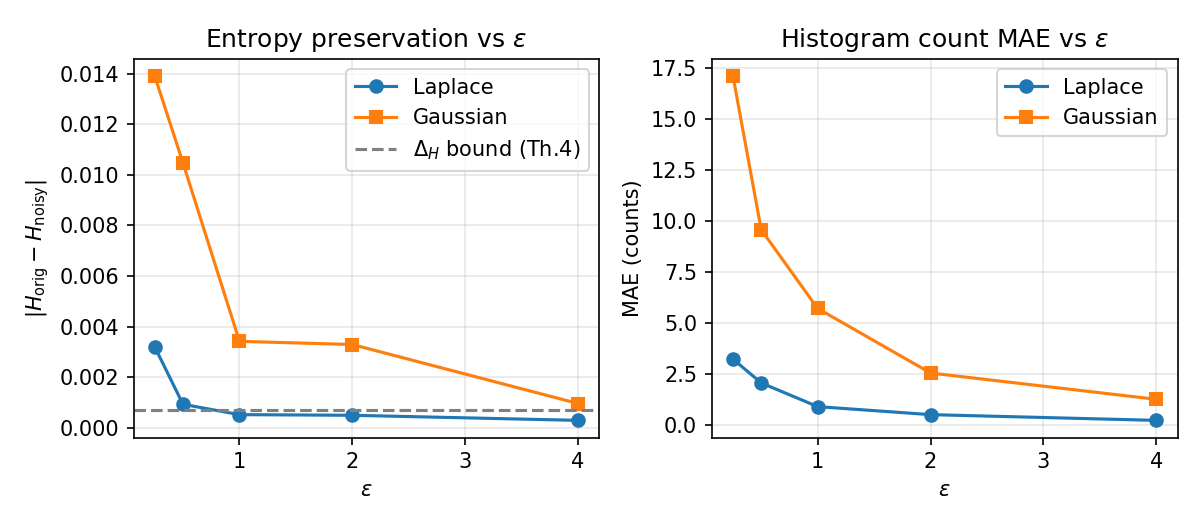}
	\caption{Baseline comparison: entropy error and count MAE vs.\ $\varepsilon$ (Laplace, Gaussian, DP synthetic (Laplace), DP synthetic (Gaussian)); $\Delta_H$ bound shown.}\label{fig:baseline}
\end{figure}

\subsection{Ablation: contribution of parameters}
Figure~\ref{fig:ablation} shows entropy error vs.\ number of histogram bins $m$ for fixed $\varepsilon=1$. The mean $|H_{\mathrm{orig}} - H_{\mathrm{noisy}}|$ remains below the $\Delta_H$ bound across bins, so the bound is useful for calibration regardless of discretization. Together with Figure~\ref{fig:generalpic} (varying $k$, $\gamma$) and Figure~\ref{fig:GaussianKernel1} (varying $\rho$), these results illustrate how $k$ and $\rho$ affect the privacy--utility tradeoff and show that the entropy bound remains informative across the tested values of $m$.
\begin{figure}[t]
	\centering
	\includegraphics[width=\columnwidth]{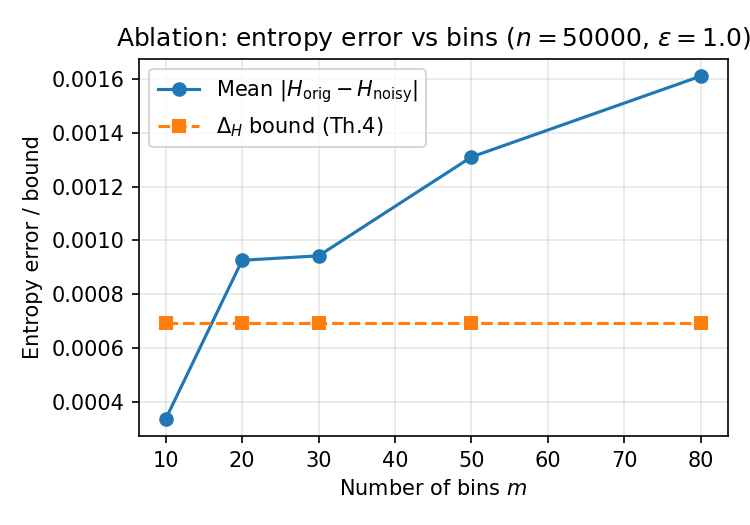}
	\caption{Ablation: entropy error vs.\ number of bins $m$ ($\varepsilon=1$); $\Delta_H$ bound.}\label{fig:ablation}
\end{figure}

\subsection{Membership inference attack evaluation}
Figure~\ref{fig:mia} reports membership-inference and linkage-style attack behavior under the tested DP histogram release setting. \textbf{MIA:} Features are entropy and total count of the release; the classifier is logistic regression. We report accuracy and AUC vs.\ $\varepsilon$, with mean and 95\% confidence interval over $5$ repeated runs. \textbf{Linkage-style attack:} The adversary sees one release and the two candidate releases (from $D$ and $D'$); the guess is which reference is closer in L2 distance. As $\varepsilon$ decreases, both MIA accuracy/AUC and linkage-style accuracy move toward random-guess performance, which is consistent with stronger privacy under smaller $\varepsilon$ in the tested setting.
\begin{figure}[t]
	\centering
	\includegraphics[width=\columnwidth]{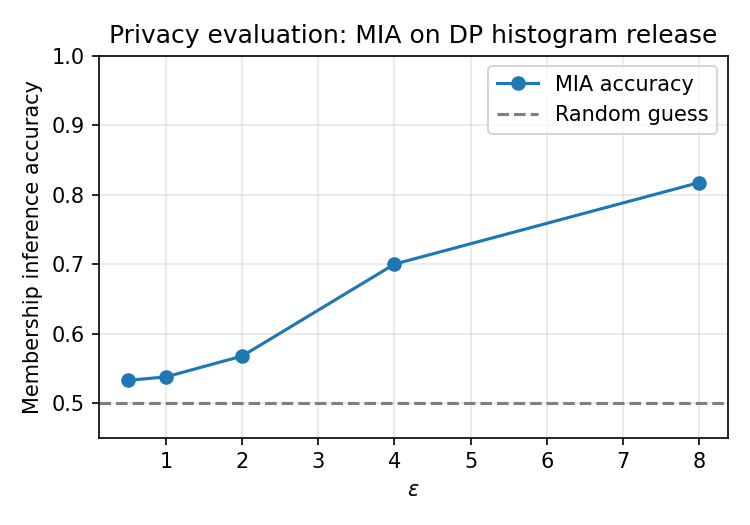}
	\caption{MIA accuracy (with 95\% CI), MIA AUC, and linkage-style accuracy vs.\ $\varepsilon$; random guess = 0.5.}\label{fig:mia}
\end{figure}

Supplementary image-noise utility results are provided in the supplementary appendix (Appendices~J--L) and are included only as supporting evidence, not as a central part of the main theorem-to-mechanism evaluation pipeline.

\subsection{Real-data experiment: histogram entropy on CSV data}

To demonstrate the framework on real tabular data, we use the \emph{Amazon Product and Google Locations Reviews} dataset \cite{adler2026amazon}, which provides time series of hourly review counts per category; the numeric column $y$ is the number of reviews per hour. We build a histogram over $y$, compute the Shannon entropy $H$ of the empirical distribution, then release a differentially private version of the histogram by adding Laplace noise to the bin counts (sensitivity 1, privacy parameter $\varepsilon$). Figure~\ref{fig:csv_entropy} shows the original and private histograms and the corresponding entropy values. Table~\ref{tab:csv_entropy} gives the numerical summary for one run.

\begin{table*}[htbp]
\centering
\caption{Real-data experiment: one run (dataset, $n$, bins, $\varepsilon$, $H_{\mathrm{orig}}$, $H_{\mathrm{noisy}}$, $\Delta_H$ bound).}
\label{tab:csv_entropy}
\begin{tabular}{lrrrrrr}
\hline
Dataset & $n$ & bins & $\varepsilon$ & $H_{\mathrm{orig}}$ & $H_{\mathrm{noisy}}$ & $\Delta_H$ bound \\
\hline
y\_amazon-google-large & 100000 & 30 & 1.0 & 3.0279 & 3.0286 & 0.00037 \\
\hline
\end{tabular}
\end{table*}

\begin{figure}[htbp]
	\centering
	\includegraphics[width=\columnwidth]{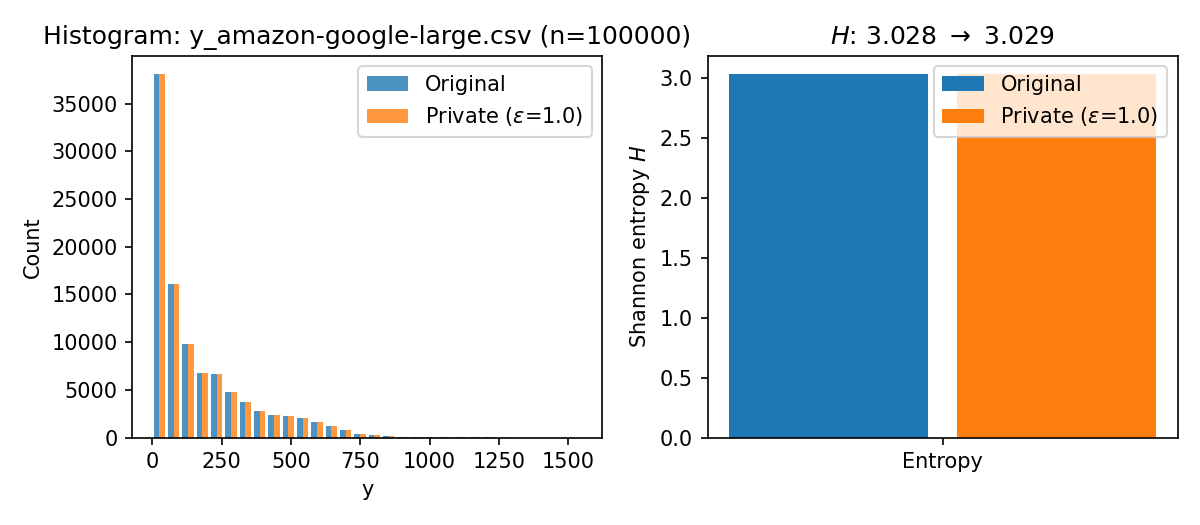}
	\caption{Real-data: histogram of $y$ (left: original vs.\ $\varepsilon$-DP counts; right: $H_{\mathrm{orig}}$ vs.\ $H_{\mathrm{noisy}}$).}\label{fig:csv_entropy}
\end{figure}

The left panel compares original bin counts with Laplace-noised counts; the right panel shows $H_{\mathrm{orig}}$ and $H_{\mathrm{noisy}}$ (Table~\ref{tab:csv_entropy}). The private histogram preserves the rough shape while satisfying $\varepsilon$-DP; Theorem~\ref{thm4} and sensitivity 1 for bin counts justify the noise scale. This illustrates the applicability of the entropy-based DP release pipeline to real tabular and time-series data in the tested setting.

\subsection{Multi-dataset validation: generality of the entropy--DP framework}
To demonstrate that the entropy-based DP release is not limited to a single dataset, we run the same pipeline (histogram over a numeric column, Laplace noise with sensitivity 1 and $\varepsilon=1$, Shannon entropy before/after) on seven public tabular datasets from different domains: Amazon-Google reviews (UCI), House Prices (Kaggle), Home Credit (Kaggle), Tabular Feature Engineering (Kaggle), synthetic regression sets (\texttt{r\_diff\_train}, \texttt{pow\_train}), and U.S.\ unemployment (UNRATE). Across all datasets, $H_{\mathrm{noisy}}$ remains close to $H_{\mathrm{orig}}$ and the $\Delta_H$ bound (Theorem~\ref{thm4}) is respected. The full numerical table and four-panel figure for multi-dataset validation are given in the supplementary appendix (Appendix~\ref{app:multidata}). An additional experimental figure is provided in Appendix~\ref{app:extrafig}. A summary of each figure's data and the theory it supports is given in Appendix~\ref{app:figdata}.

\section{Discussion}

The experimental results support the main claims of the paper at three levels. First, for histogram-based release, the empirical entropy change remains below the theoretical sensitivity bound in the tested range, which is consistent with the intended role of Theorem~\ref{thm4} as a calibration tool rather than merely a descriptive quantity. Second, the baseline comparisons indicate that standard Laplace and Gaussian perturbations follow similar utility trends under the same privacy regime, while remaining compatible with the entropy-based analysis used in this work. Third, the attack-based experiments show that membership-inference and linkage-style attack performance weakens as $\varepsilon$ decreases, which is consistent with the expected privacy--utility tradeoff under stronger privacy protection.

At the same time, these results should be interpreted with appropriate scope. The experimental evidence supports the proposed framework in the tested histogram and tabular-release settings, but it does not by itself establish superiority over all differentially private synthetic-data methods or all privacy attacks. In particular, the present attack evaluation is limited to the implemented membership-inference and linkage-style protocols, while attribute inference is included in the threat model but not benchmarked experimentally. For this reason, the results are best interpreted as evidence that the proposed release pipeline behaves consistently with its formal guarantees and intended privacy objectives in the tested setting.

The Wiener-kernel component provides an additional perspective on privacy-preserving release beyond finite-dimensional histograms. The experiments suggest that increasing the penalty parameter $\rho$ improves agreement between the private and non-private RKHS means, illustrating a controllable privacy--utility tradeoff in the functional-space setting. However, this part of the framework also depends on finite-dimensional numerical approximation of an infinite-dimensional construction, so its current role is better viewed as a structured extension of the main release framework rather than a complete standalone theory of infinite-dimensional private release.

The synthetic-data mechanism $\mathcal{F}$ should likewise be interpreted carefully. The theorem establishes an $(\varepsilon,\delta)$-DP guarantee under the stated mechanism definition and parameter conditions, and the privacy-test experiment illustrates that the mechanism can remain usable under nontrivial threshold settings. Nevertheless, practical deployment quality still depends on the choice of proposal distribution, partition structure, and parameter tuning. Accordingly, the present paper should be read as establishing a formal and experimentally supported mechanism class in the tested regime, rather than as resolving all practical design choices for synthetic-data generation.

More broadly, the results highlight a useful methodological point: privacy analysis benefits from combining theorem-level calibration, standard differential privacy baselines, and explicit attack-based evaluation. Any one of these elements in isolation gives only a partial picture. Theoretical bounds alone do not show practical behavior; empirical attacks alone do not provide worst-case guarantees; and utility metrics alone do not reveal privacy leakage. By combining these views, the paper aims to provide a more interpretable release framework for sensitive data publication.

Several limitations remain. The histogram-based analysis depends on discretization and may become sensitive to binning choices in high-dimensional settings. The current synthetic-data experiments focus on tabular data, and broader validation on multimodal, sequential, or federated data remains open. The attack suite can also be strengthened by including attribute-inference benchmarks and stronger adversarial protocols. Finally, in the Wiener-kernel setting, a tighter account of the approximation gap between the infinite-dimensional formulation and its finite-dimensional implementation would further strengthen the theoretical side of the paper.

\section*{Conclusion}
We have presented REAEDP, a differential privacy framework for entropy-calibrated data release with formal analysis and empirical evaluation. The main technical results are an explicit sensitivity bound for histogram entropy (Theorem~\ref{thm4}) and an $(\varepsilon,\delta)$-differential privacy guarantee for the synthetic-data mechanism $\mathcal{F}$ under the stated mechanism and parameter conditions (Theorem~\ref{thm5}). Across multiple public tabular datasets, the experiments indicate that the empirical entropy change remains below the theoretical bound in the tested regime, that standard Laplace and Gaussian baselines show comparable trends, and that membership-inference and linkage-style attack performance weakens as $\varepsilon$ decreases. Taken together, these results support REAEDP as a practically usable privacy-preserving release framework in the tested settings.

\textbf{Limitations and future work.} The entropy-based release studied here depends on histogram discretization and may be sensitive to bin selection, especially in higher-dimensional settings. Extending the method to very high-dimensional or very large-scale data may require approximate, adaptive, or distributed implementations. In the Wiener-kernel analysis, the infinite-dimensional formulation is implemented numerically through finite-dimensional approximation, and a tighter treatment of the corresponding approximation error under differential privacy would further strengthen the theory. In addition, although the paper explicitly evaluates membership inference and a linkage-style attack, broader experimental coverage of attribute inference, stronger attack models, and more diverse multimodal or federated settings remains for future work.

\input{appendix_content}
\bibliographystyle{IEEEtran}
\bibliography{IEEEabrv,ref}
\end{document}

%% file: appendix_content.tex
\appendix
\renewcommand{\thefigure}{S\arabic{figure}}
\renewcommand{\thetable}{S\arabic{table}}
\setcounter{figure}{0}
\setcounter{table}{0}

\section{Mechanism $\mathcal{F}$: formal specification}
We recall the synthetic-data mechanism $\mathcal{F}$ used in the main text and provide a theorem-level specification of its components.

\subsection*{Domain and range}
Let $D=\{x_1,\dots,x_n\}$ be a dataset over universe $\mathcal{U}$. The mechanism $\mathcal{F}$ is a randomized map that outputs a synthetic record $y \in \mathcal{U}$.

\subsection*{Input}
The mechanism takes as input the dataset $D$ together with tuning parameters $k,t,\gamma,\varepsilon_0$ as used in Theorem~\ref{thm5}.

\subsection*{Seed selection}
A seed $s$ is drawn uniformly at random from $D$:

\[
s \sim \mathrm{Unif}(D).
\]
\subsection*{Candidate generation}
Conditional on $s$, a candidate $y \in \mathcal{U}$ is drawn from a seed-dependent proposal distribution $p_s(\cdot)$:
\[
y \sim p_s(\cdot), \qquad p_s(y) \ge 0,\ \sum_{y\in\mathcal{U}} p_s(y) \le 1.
\]
The definitions and bounds in Lemmas~\ref{lemma:2yinU}--\ref{lemma5:random1} are expressed in terms of $p_s(y)$.
\subsection*{Partition rule}
For each seed $s$ and candidate $y$, an index
\[
I_s(y) \in \{0,1,2,\dots\}
\]
is computed, and for dataset $D$ we define the corresponding cell
\[
C_j(D,y) = \{s \in D : I_s(y) = j\},
\]
with cardinality $|C_j(D,y)|$. The quantity $c(d',y)=|C_{I_{d'}(y)}(D,y)|$ used in Theorem~\ref{thm5} is defined through this partition.

\subsection*{Privacy-test random variable}
Let $L$ be an integer-valued random variable, independent of the seed-selection and candidate-sampling steps, and assume its tail probabilities satisfy

\[
\mathrm{pt}(D,j,y) \stackrel{\mathrm{def}}{=} \mathbb{P}\mathrm{r}\{L \ge k - |C_j(D,y)|\}
\]
satisfies the inequalities in Lemmas~\ref{lemma4:djpartition} and~\ref{lemma5:random1}. In particular, $\varepsilon_0>0$ is chosen so that
\[
\mathbb{P}\mathrm{r}\{L \ge k - |C_j(D',y)| + 1\} \le e^{\varepsilon_0}\mathbb{P}\mathrm{r}\{L \ge k - |C_j(D',y)|\}
\]
for all relevant $(D,D',j,y)$.
\subsection*{Acceptance rule}
Given $(D,s,y)$ and the partition index $j=I_s(y)$, the privacy test accepts the candidate with probability

\[
\mathrm{Accept}(D,s,y) = \mathbf{1}\{L \ge k - |C_j(D,y)|\},
\]
where $L$ is drawn as above. Equivalently, the probability that $y$ generated from $D$ and a uniformly sampled seed passes the test is
\[
q(D,j,y) = \mathrm{pt}(D,j,y) \sum_{s\in C_j(D,y)} p_s(y),
\]
as used in Lemma~\ref{lemma:2yinU}.
\subsection*{Output and termination}
The mechanism $\mathcal{F}$ repeatedly performs:
\begin{enumerate}
    \item sample $s\sim \mathrm{Unif}(D)$;
    \item sample $y\sim p_s(\cdot)$;
    \item compute $j = I_s(y)$ and draw $L$;
    \item if $L \ge k - |C_j(D,y)|$, output $y$ and terminate; otherwise, resample $(s,y,L)$.
\end{enumerate}
The proofs of Lemmas~\ref{lemma:2yinU}--\ref{lemma:neighboringdatasets} and Theorem~\ref{thm5} analyse the distribution of $\mathcal{F}(D)$ under this procedure. In practice, one can bound the expected number of resampling steps using the pass-rate plots in Figure~\ref{fig:generalpic} and choose $(k,t,\gamma,\varepsilon_0)$ so that the rejection probability is small while maintaining the $(\varepsilon,\delta)$-DP guarantee of Theorem~\ref{thm5}.

\section{Proofs}\label{app:proofs}

\subsection{Proof of Theorems \ref{thm1} and \ref{thm2}}
Sequential and advanced composition are standard; we refer to~\cite{dwork2014algorithmic}.

\subsection{Proof of Theorem \ref{thm4} (Shannon entropy sensitivity under replacement adjacency)}
\begin{proof}
Let $z$ and $z'$ be the random variables associated with the two histograms (of $m$ bins) from two adjacent datasets $D$ and $D'$, respectively; both have $n$ records and differ in exactly one record under the replacement adjacency model of Section~III-A.

Let $(c_1,\ldots,c_m)$ and $(c'_1,\ldots,c'_m)$ denote the histogram counts for $z$ and $z'$, respectively. Entropy is calculated on the probability distribution represented by the histogram.

Note that the histograms of the adjacent datasets $D$ and $D'$ can differ in at most two positions. If there is no difference in any position, then $\Delta_H = 0$ and the bound holds trivially. Without loss of generality, there exist indices $j_1 \neq j_2$ such that $c'_{j_1} = c_{j_1} + 1$ and $c'_{j_2} = c_{j_2} - 1$.

Moreover, $n-1 \geq c_{j_1} \geq 0$, so $n \geq c'_{j_1} \geq 1$; and $n \geq c_{j_2} \geq 1$, so $n-1 \geq c'_{j_2} \geq 0$. For $i \neq j_1, j_2$, we have $c'_i = c_i$, and

\begin{eqnarray}
\sum^m_{i=1} c_i = \sum^m_{i=1} c'_i =n
\end{eqnarray}

We have
\begin{equation}
\begin{aligned}
H(z) = - \sum^m_{i=1} \frac{c_i}{n} \log_2 \frac{c_i}{n} \\
    =  - \frac{1}{n} \Bigl[\sum^m_{i=1} c_i \log_2 c_i - n \log_2 n\Bigr] \\
    = \log_2 n - \frac{1}{n} \Bigl( c_{j_1} \log_2 c_{j_1} + c_{j_2} \log_2 c_{j_2}  \\
       + \sum_{i \neq j_1,j_2} c_i \log_2 c_i \Bigr)
\end{aligned}
\end{equation}
Similarly,
\begin{eqnarray}
\begin{aligned}
H(z') = \log_2 n - \frac{1}{n} \sum_{i \neq j_1,j_2} c_i \log_2 c_i \\
 - \frac{1}{n}\Bigl[(c_{j_1} + 1)\log_2 (c_{j_1} +1) + (c_{j_2} - 1) \log_2 (c_{j_2} -1)\Bigr].
 \end{aligned}
\end{eqnarray}

We have $\Delta_H = \max_{c_{j_1}, c_{j_2}} |H(z)-H(z')|$. For brevity we omit the maximum and analyze the relationship between $c_{j_1}$ and $c_{j_2}$ to show the bound in each case. We have
	\begin{eqnarray}
	\begin{aligned}
	\Delta_H =  |H(z) - H(z')|  \\
			= \frac{1}{n}\Bigl|c_{j_1}\log_2 c_{j_1} - (c_{j_1}+1) \log_2 (c_{j_1}+1) \\
			 + c_{j_2}\log_2 c_{j_2} - (c_{j_2} -1)\log_2(c_{j_2}-1)\Bigr|.
			  \end{aligned}
	\end{eqnarray}	
	In the first case, $c_{j_1} = 0$. We have
	 \begin{eqnarray}
	  \Delta_H = \frac{1}{n}\Bigl|c_{j_2}\log_2 c_{j_2}  - (c_{j_2}-1) \log_2 (c_{j_2}-1)\Bigr|.
	  \end{eqnarray}
	 If $c_{j_2} = 1$, then $\Delta_H = 0$ and the bound holds. So suppose $c_{j_2}> 1$. We have
	  	  	\begin{eqnarray}
	  	  	\begin{aligned}	
	  \Delta_H = \frac{1}{n}\Bigl|c_{j_2}\log_2 c_{j_2} - (c_{j_2}-1) \log_2 (c_{j_2}-1)\Bigr| \\
	  =  \frac{1}{n} \Bigl|c_{j_2}\log_2\Bigl(\frac{c_{j_2}}{c_{j_2}-1}\Bigr) + \log_2(c_{j_2}-1)\Bigr| \\
	  \leq \frac{1}{n} \Bigl|c_{j_2}\log_2\Bigl(\frac{c_{j_2}}{c_{j_2}-1}\Bigr)\Bigr| + \frac{1}{n} \log_2 (c_{j_2}-1) \\
	  \leq  \frac{1}{n} \log_2(n-1) + \frac{1}{n} \Bigl|(a+1) \log_2 \Bigl(1+ \frac{1}{a}\Bigr)\Bigr|.
	  			  \end{aligned}
	  \end{eqnarray}
	Here $(a + 1) \log_2(1 +\frac{1}{a})\leq 2$. We conclude $\Delta_H \leq \frac{1}{n}(2 + \log_2 n)$.

	 In the second case, $c_{j_1} = 1$. We have
	 \begin{eqnarray}
	 \Delta_H = \frac{1}{n}\Bigl|c_{j_1}\log_2 c_{j_1}  - (c_{j_1}+1) \log_2 (c_{j_1}+1)\Bigr|.
	 \end{eqnarray}
	 If $c_{j_2} = 0$, then the bound holds trivially. So suppose $c_{j_2} \geq 1$. We have
	 \begin{eqnarray}
	 \begin{aligned}	
	 \Delta_H = \frac{1}{n}\Bigl|c_{j_1}\log_2 c_{j_1} - (c_{j_1}+1) \log_2 (c_{j_1}+1)\Bigr| \\
	 =  \frac{1}{n} \Bigl|c_{j_1}\log_2\Bigl(\frac{c_{j_1}}{c_{j_1}+1}\Bigr) - \log_2(c_{j_1}+1)\Bigr| \\
	 \leq \frac{\log_2 n}{n} + \frac{1}{n}\Bigl| c_{j_1} \log_2 \Bigl(\frac{c_{j_1}}{c_{j_1}+1}\Bigr)\Bigr| \\
	 = \frac{\log_2 n}{n}  + \frac{1}{n} c_{j_1} \log_2 \Bigl(\frac{c_{j_1}+1}{c_{j_1}}\Bigr) \\
	  = \frac{\log_2 n}{n}  + \frac{1}{n} c_{j_1} \log_2 \Bigl(1+ \frac{1}{c_{j_1}}\Bigr).
	 \end{aligned}
	 \end{eqnarray}
	 By L'H\^{o}pital's rule, $c_{j_1} \log_2(1 + \frac{1}{c_{j_1}}) \leq \frac{1}{\ln 2}$. Thus $\Delta_H \leq \frac{1}{n}(2+ \frac{1}{\ln 2} + 2\log_2 n)$.

	 In the third case, $c_{j_1} \geq 1$ and $c_{j_2} \geq 1$. We have
	\begin{eqnarray}
	 \begin{aligned}	
	 \Delta_H &= \frac{1}{n}\Bigl|c_{j_1}\log_2 c_{j_1} - (c_{j_1}+1)\log_2(c_{j_1}+1) \\
	 &\qquad {}+ c_{j_2}\log_2 c_{j_2} - (c_{j_2}-1)\log_2(c_{j_2}-1)\Bigr| \\
	 &\leq \frac{1}{n}\Bigl|c_{j_1}\log_2 c_{j_1} - (c_{j_1}+1)\log_2(c_{j_1}+1)\Bigr| \\
	 &\quad {}+ \frac{1}{n}\Bigl|c_{j_2}\log_2 c_{j_2} - (c_{j_2}-1)\log_2(c_{j_2}-1)\Bigr|.
	 \end{aligned}
	 \end{eqnarray}
	 
	 For Case 1 and Case 2, each term on the right-hand side is bounded. Combining them, we conclude $\Delta_H \leq \frac{1}{n} (2+ \frac{1}{\ln 2} + 2 \log_2 n)$.
\end{proof}

\subsection{R\'enyi entropy sensitivity (explicit upper bound)}
Under the same adjacency model as in Theorem~\ref{thm4}, the histograms $(c_1,\ldots,c_m)$ and $(c'_1,\ldots,c'_m)$ differ in at most two positions $j_1,j_2$ with $c'_{j_1}=c_{j_1}+1$, $c'_{j_2}=c_{j_2}-1$, and $\sum_i c_i = \sum_i c'_i = n$. Let $p_i = c_i/n$, $p'_i = c'_i/n$, $S = \sum_{i=1}^m p_i^\alpha$, and $S' = \sum_{i=1}^m (p'_i)^\alpha$. Then $H_\alpha(z) = \frac{1}{1-\alpha}\log_2 S$ and $H_\alpha(z') = \frac{1}{1-\alpha}\log_2 S'$ for $\alpha>0$, $\alpha\neq 1$.

\paragraph{Step 1: bound $|S'-S|$.} Only bins $j_1$ and $j_2$ change: $p'_{j_1} = p_{j_1} + 1/n$, $p'_{j_2} = p_{j_2} - 1/n$. By the mean value theorem, $(p+1/n)^\alpha - p^\alpha = \alpha \xi^{\alpha-1}/n$ for some $\xi \in (p, p+1/n)$. For $\alpha \geq 1$, $\xi^{\alpha-1} \leq 1$; for $0<\alpha<1$, $\xi^{\alpha-1} \leq 1$ on $(0,1]$. So $|(p+1/n)^\alpha - p^\alpha| \leq |\alpha|/n$. Similarly $|(p-1/n)^\alpha - p^\alpha| \leq |\alpha|/n$. Hence
\begin{equation}
|S' - S| = \bigl| \bigl((p_{j_1}+\tfrac{1}{n})^\alpha - p_{j_1}^\alpha\bigr) + \bigl((p_{j_2}-\tfrac{1}{n})^\alpha - p_{j_2}^\alpha\bigr) \bigr| \leq \frac{2|\alpha|}{n}.
\end{equation}

\paragraph{Step 2: bound $|\log_2 S' - \log_2 S|$.} For any probability vector $(p_i)$ and $\alpha > 0$, $\alpha \neq 1$, we have $S = \sum_i p_i^\alpha \geq 1$ (by the power-mean inequality: for $\alpha \in (0,1)$, $\sum p_i^\alpha \geq (\sum p_i)^\alpha = 1$; for $\alpha > 1$, $\sum p_i^\alpha \geq \sum p_i = 1$). So $1/S \leq 1$. Hence $\log_2 S' - \log_2 S = \log_2(1 + (S'-S)/S)$ satisfies $|\log_2(1+x)| \leq |x|/\ln 2$ for small $|x|$, so
\begin{equation}
|\log_2 S' - \log_2 S| \leq \frac{1}{\ln 2}\,\frac{|S'-S|}{S} \leq \frac{2|\alpha|}{n \ln 2}.
\end{equation}

\paragraph{Step 3: explicit sensitivity bound.} Therefore
\begin{equation}
\Delta_{H_\alpha} \stackrel{\mathrm{def}}{=} \max_{D \sim D'} |H_\alpha(z) - H_\alpha(z')| \leq \frac{2|\alpha|}{|1-\alpha| \ln 2}\,\frac{1}{n}.
\end{equation}
For fixed $\alpha \neq 1$, $\Delta_{H_\alpha} = O(1/n)$, i.e., the same order as the Shannon sensitivity bound $\Delta_H$ in Theorem~\ref{thm4}. Thus the R\'enyi entropy sensitivity admits an explicit upper bound, and the title and abstract's mention of R\'enyi entropy is consistent with the theory.

\paragraph{Varying $\alpha$: calibration and risk.} To give R\'enyi entropy an operational role in calibration, we report how the sensitivity bound $\Delta_{H_\alpha}$ depends on $\alpha$. The factor $2|\alpha|/(|1-\alpha|\ln 2)$ is minimized near $\alpha=1$ (where it coincides with the Shannon regime) and grows as $\alpha$ moves away from 1. Table~\ref{tab:varying_alpha} gives the constant factor $C_\alpha = 2|\alpha|/(|1-\alpha|\ln 2)$ for selected $\alpha$, so that $\Delta_{H_\alpha} \leq C_\alpha/n$. For $\alpha>1$, $C_\alpha$ increases with $\alpha$, so stricter R\'enyi orders require slightly larger noise scale for the same $n$; for $\alpha<1$, $C_\alpha$ is larger than at $\alpha=1$, so calibration using $H_\alpha$ can be done with the same $1/n$ scaling and a $\alpha$-dependent constant. This supports using the R\'enyi bound in the title and method: one can choose $\alpha$ (e.g., for R\'enyi DP or for robustness) and calibrate release via $\Delta_{H_\alpha}$.

\begin{table}[htbp]
\centering
\caption{Varying $\alpha$: sensitivity constant $C_\alpha$ such that $\Delta_{H_\alpha} \leq C_\alpha/n$.}
\label{tab:varying_alpha}
\begin{tabular}{c|cccccc}
\hline
$\alpha$ & 0.5 & 1.5 & 2 & 3 & 5 & 10 \\
\hline
$C_\alpha$ & 2.89 & 8.66 & 5.77 & 4.33 & 3.61 & 3.21 \\
\hline
\end{tabular}
\end{table}

\subsection{Proof of Lemma \ref{lemma:2yinU}}
\begin{proof}
Fix $y$. Following the description of the mechanism, we have
\begin{equation}
\begin{aligned}
\mathbb{P}\mathrm{r}\{\mathcal{F}(D^*)=y\} = \sum_{s \in D^*} \mathbb{P}\mathrm{r}\{\mathrm{seed}\;s,\, (D^*,s,y)\;\mathrm{pass\;test}\} \\
=\dfrac{1}{|D^*|} \sum_{s\in D^*} \bigl[p_s(y)\,\mathbb{P}\mathrm{r}\{(D^*,s,y)\;\mathrm{pass\;test}\}\bigr]
\end{aligned}
\end{equation}
\end{proof}

\subsection{Proof of Lemma \ref{lemma4:djpartition} and Corollary \ref{corollary:ismallthan0}}
\begin{proof}
	Fix $y$ and let $j$ be the partition into which $d'$ falls.

    For part (a), for $i \neq j$ we have $C_i(D,y) = C_i(D',y)$, so $q(D,i,y) = q(D',i,y)$.

    For part (b), we have $p_{d'}(y)>0$. By Lemma~\ref{lemma4:djpartition}, $\mathrm{pt}(D',j,y)\leq \mathrm{pt}(D,j,y)$. Thus
			\begin{equation}
	\begin{aligned}
	q(D, j, y) = \mathrm{pt}(D, j, y) \sum_{s\in C_j(D,y)} p_s(y)  \\
	< \mathrm{pt}(D, j, y) \Bigl[ \sum_{s\in C_j(D,y)} p_s(y) + p_{d'}(y)  \Bigr] \\
	= \mathrm{pt}(D, j, y) \sum_{s\in C_j(D',y)} p_s(y) \\
	\leq  \mathrm{pt}(D', j, y) \sum_{s\in C_j(D',y)} p_s(y) \\
	= q(D', j, y).
		\end{aligned}
	\end{equation}

	If $|C_j(D,y)|<t$, we have
			\begin{equation}
\begin{aligned}
q(D', j, y) = \mathrm{pt}(D', j, y) \sum_{s\in C_j(D,y)} p_s(y)  \\
\leq e^{-\varepsilon_0(k-t)}\sum_{s\in C_j(D,y)} p_s(y) \\
\leq t e^{-\varepsilon_0(k-t)},
\end{aligned}
\end{equation}
	since $\mathrm{pt}(D',j,y)= \mathbb{P}\mathrm{r}\{L\geq k-|C_j(D',y)|\}\leq \mathbb{P}\mathrm{r}\{L \geq k-t\}$ and $p_d(y)\leq 1$ for any $d$.

	If $|C_j(D,y)| \geq t$, we have
			\begin{equation}
\begin{aligned}
q(D', j, y) = \mathrm{pt}(D', j, y) \sum_{s\in C_j(D,y)} p_s(y)  \\
=  \mathrm{pt}(D', j, y) \Bigl[\sum_{s\in C_j(D,y)} p_s(y) +p_{d'}(y)\Bigr]  \\
\leq e^{\varepsilon_0}\mathrm{pt}(D, j, y) \Bigl[\sum_{s\in C_j(D,y)} p_s(y) +p_{d'}(y)\Bigr]  \\ 
\leq e^{\varepsilon_0} \Bigl[ 1+\frac{\gamma}{t}\Bigr]\mathrm{pt}(D, j, y) \sum_{s\in C_j(D,y)} p_s(y)   \\ 
\leq e^{\varepsilon_0} q(D, j, y),
\end{aligned}
\end{equation}
	where the last step uses Lemma~\ref{lemma4:djpartition} and the bound $p_{d'}(y)\leq \frac{\gamma}{t} \sum_{s \in C_j(D,y)} p_s(y)$.
\end{proof}

\subsection{Proof of Lemma \ref{lemma5:random1}}
\begin{proof}
	Fix an arbitrary synthetic record $y \in \mathcal{U}$ and any dataset $D$ with $|D| \geq k$. Let $D'=D \cup \{d'\}$ for some $d' \in \mathcal{U}$. By Lemma~\ref{lemma:2yinU} applied to $D$, we have $\mathbb{P}\mathrm{r}\{\mathcal{F}(D)=y\} = \frac{1}{|D|} \sum_{i \geq 0} q(D,i,y)$. By Corollary~\ref{corollary:ismallthan0}, $q(D,i,y)\leq q(D',i,y)$ for all $i$. Thus:
	
	\begin{equation}
	\begin{aligned}
	\mathbb{P}\mathrm{r}\{ \mathcal{F}(D) = y\}  = \frac{1}{|D|}\sum_{i\geq 0} q(D,i,y)  \\
	\leq \frac{1}{|D|}\sum_{i\geq 0} q(D',i,y)  \\
	= \frac{|D'|}{|D|} \mathbb{P}\mathrm{r}\{ \mathcal{F}(D') = y\} \\
	\leq (1+ \frac{1}{k}) \mathbb{P}\mathrm{r}\{ \mathcal{F}(D') = y\}
	\end{aligned}
	\end{equation}
	By assumption, $\gamma > 1$ and $t \leq k$, so $\frac{1}{k} \leq \frac{1}{t}$ and $1 + \frac{1}{k} < 1 +\frac{\gamma}{t} \leq e^{\varepsilon_0}(1 +\frac{\gamma}{t})=e^\varepsilon$. This proves the first part.

	For the second part, apply Lemma~\ref{lemma:2yinU} to $D'$ and let $j = I_{d'}(y)$ be the partition index of $d'$. We have:
		\begin{equation}
	\begin{aligned}
	\mathbb{P}\mathrm{r}\{ \mathcal{F}(D) = y\}  = \frac{1}{|D|}\sum_{i\geq 0} q(D,i,y)  \\
	= \frac{1}{|D|}\Bigl[\sum_{i\geq 0:i\neq j} q(D',i,y)+q(D',j,y)\Bigr]  \\
	= \frac{1}{|D|}\sum_{i\geq 0:i\neq j} q(D,i,y)+\frac{q(D',j,y)}{|D'|}.
	\end{aligned}
	\end{equation}
	The final equality follows from Lemma~\ref{lemma4:djpartition} part (a).

	Using Lemma~\ref{lemma4:djpartition} part (b), we consider two cases.

	\emph{Case 1:} $|C_j(D,y)| < t$. We have:
		\begin{equation}
\begin{aligned}
\mathbb{P}\mathrm{r}\{ \mathcal{F}(D') = y\}  = \frac{1}{|D|}\sum_{i\geq 0:i\neq j} q(D,i,y)+\frac{q(D',j,y)}{|D'|} \\
\leq \frac{1}{|D'|} \sum_{i\geq 0:i\neq j} q(D,i,y)+ \delta (D',j,y) \\
\leq \frac{1}{|D'|} \sum_{i\geq 0} q(D,i,y)+ \delta (D',j,y) \\
=  \frac{|D|}{|D'|}\mathbb{P}\mathrm{r}\{ \mathcal{F}(D) = y\} + \delta (D',j,y) \\
\leq \mathbb{P}\mathrm{r} \{ \mathcal{F} (D) = y\} + \delta
\end{aligned}
\end{equation}

where $\delta(D',j,y) = \frac{1}{|D'|} e^{-\varepsilon_0(k-t)}\sum_{s\in C_j(D',y)}p_s(y)$.

	\emph{Case 2:} $|C_j(D,y)| \geq t$. We have:
			\begin{equation}
	\begin{aligned}
	\mathbb{P}\mathrm{r}\{ \mathcal{F}(D') = y\}  \\
	= \frac{1}{|D'|}[\sum_{i\geq 0:i\neq j} q(D,i,y)+q(D',j,y)] \\
	\leq \frac{1}{|D'|}[\sum_{i\geq 0:i\neq j} q(D,i,y)+e^{\varepsilon_0}[1+\frac{\gamma}{t}]q(D,j,y)] \\
	\leq e^{\varepsilon_0} [1+\frac{\gamma}{t}]\frac{1}{|D'|}\sum_{i\geq 0:i\neq j} q(D,i,y)\\
	= e^{\varepsilon_0} [1+\frac{\gamma}{t}]\frac{|D|}{|D'|} \mathbb{P}\mathrm{r}\{ \mathcal{F}(D) = y\}\\
	\leq e^{\varepsilon_0}  \mathbb{P}\mathrm{r}\{ \mathcal{F}(D) = y\}
	\end{aligned}
	\end{equation}
	For $\varepsilon = \varepsilon_0 + \ln(1 +\frac{\gamma}{t})$, we have $\frac{|D|}{|D'|} < 1$, so $\mathbb{P}\mathrm{r}\{\mathcal{F}(D') = y\} \leq e^{\varepsilon} \mathbb{P}\mathrm{r}\{\mathcal{F}(D) = y\}$.

	Letting $\delta(D',d',y)=\delta(D',I_{d'}(y),y)$, we obtain the claim. This completes the proof.
\end{proof}

\subsection{Proof of Lemma \ref{lemma:neighboringdatasets} (Neighboring datasets)}\label{Lemma:proofofneighbor}
\begin{proof}
	There are two cases: $i = I_{d'}(y)$ or $i \neq I_{d'}(y)$. If $i = I_{d'}(y)$, then $d'$ falls into partition $i$, so $C_i(D',y) = C_i(D,y) \cup \{d'\}$. We have $|C_i(D,y)| = |C_i(D',y)| - 1$, so
		\begin{equation}
	\begin{aligned}
	\mathrm{pt}(D, i, y)
	&= \mathbb{P}\mathrm{r} \{ L \geq k - |C_i(D,y)|\} \\
	&= \mathbb{P}\mathrm{r} \{ L \geq k - |C_i(D',y)| + 1\}  \\
	&\leq e^{\varepsilon_0} \mathbb{P}\mathrm{r} \{ L \geq k - |C_i(D',y)|\} \\
	&= e^{\varepsilon_0} \mathrm{pt}(D',i,y).
	\end{aligned}
	\end{equation}

	If $i \neq I_{d'}(y)$, then $C_i(D',y) = C_i(D,y)$, so
		\begin{equation}
	\mathrm{pt}(D', i, y) = \mathbb{P}\mathrm{r} \{ L \geq k - |C_i(D,y)|\}  =  \mathrm{pt}(D,i,y).
	\end{equation}
	Combining both cases yields the result.
\end{proof}

\subsection{Proof of Theorem \ref{thm5} (Differential privacy of $\mathcal{F}$)}\label{thm5:proof}
\begin{proof}
Fix a dataset $D$ with $|D|\geq k$ and any record $d' \in \mathcal{U}$; let $D'=D \cup \{d'\}$. The mechanism $\mathcal{F}$ has range $\mathcal{U}$, so all outputs are (random) elements of $\mathcal{U}$. Fix an arbitrary measurable $Y \subseteq \mathcal{U}$.

We prove that for either $(D_1,D_2)=(D,D')$ or $(D_1,D_2)=(D',D)$, we have $\mathbb{P}\mathrm{r}\{\mathcal{F}(D_1)\in Y\} \leq e^\varepsilon\,\mathbb{P}\mathrm{r}\{\mathcal{F}(D_2)\in Y\} + \delta$.

\emph{Case $D_1=D$, $D_2=D'$.} By Lemma~\ref{lemma5:random1} we obtain:
    \begin{equation}
\begin{aligned}
\mathbb{P}\mathrm{r}\{\mathcal{F}(D)\in Y\} = \sum_{y\in Y} \mathbb{P}\mathrm{r}\{\mathcal{F}(D)=y\}  \\
\leq \sum_{y\in Y} e^{\varepsilon}\,\mathbb{P}\mathrm{r}\{\mathcal{F}(D')=y\}  \\
= e^{\varepsilon}\,\mathbb{P}\mathrm{r}\{\mathcal{F}(D')\in Y\}  
\end{aligned}
\end{equation}
\emph{Case $D_1=D'$, $D_2=D$.} Define $c(d',y)=|C_{I_{d'}(y)}(D,y)|$. Partition $Y$ into $Y_{t-}= \{y \in Y : c(d',y)<t\}$ and $Y_{t+}= \{y \in Y : c(d',y)\geq t\}$, so $Y = Y_{t-} \cup Y_{t+}$. We have:
        {\footnotesize
    \begin{equation}
\begin{aligned}
\mathbb{P}\mathrm{r}\{ \mathcal{F}(D') \in Y\}  &= \sum_{ y\in Y} \mathbb{P}\mathrm{r}\{ \mathcal{F}(D') = y\}  \\
&\leq \sum_{ y\in Y_{t+}} e^{\varepsilon}\mathbb{P}\mathrm{r}\{ \mathcal{F}(D) = y\} \\
&\quad + \sum_{ y\in Y_{t-}} \bigl[ e^{\varepsilon}\mathbb{P}\mathrm{r}\{ \mathcal{F}(D) = y\} + \delta (D', I_{d'}(y),y) \bigr]  \\
&= e^{\varepsilon}\sum_{ y\in Y} \mathbb{P}\mathrm{r}\{ \mathcal{F}(D) = y\} + \sum_{ y\in Y_{t-}} \delta (D', d',y) \\
&= e^{\varepsilon} \mathbb{P}\mathrm{r}\{ \mathcal{F}(D) \in Y\} + \sum_{ y\in Y_{t-}} \delta (D', d',y). 	
\end{aligned}
\end{equation}
    }
The inequality uses Lemma~\ref{lemma5:random1} (Case 1 for $y \in Y_{t-}$, Case 2 for $y \in Y_{t+}$) and Lemma~\ref{lemma4:djpartition}.

Thus $\sum_{y \in Y_{t-}} \delta(D',d',y) \leq e^{-\varepsilon_0(k-t)}$, so the mechanism satisfies $(\varepsilon,\delta)$-DP with $\delta = e^{-\varepsilon_0(k-t)}$. With $C(D',d',y)=C_{I_{d'}(y)}(D',y)$, we have 
    \begin{equation}
\begin{aligned}
\sum_{ y\in Y_{t-}} \delta (D', d',y) 	 = \sum_{ y\in Y_{t-}} \frac{e^{-\varepsilon_0(k-t)}}{|D'|}\sum_{ s\in C(D', d',y)} p_s(y)  \\
    =	\frac{e^{-\varepsilon_0(k-t)}}{|D'|} \sum_{s \in D'} \sum_{y\in Y_{t-}} \mathbb{I}_{I_{d'}(y)=I_s(y)}\, p_s(y) \\
        \leq e^{-\varepsilon_0(k-t)}
\end{aligned}
\end{equation}
For each $s \in D'$, $\sum_{y\in Y_{t-}} \mathbb{I}_{I_{d'}(y)=I_s(y)}\, p_s(y) \leq \sum_{y\in \mathcal{U}} p_s(y) \leq 1$. Hence the right-hand side is at most $e^{-\varepsilon_0(k-t)}$. This completes the proof of Theorem~\ref{thm5}.
\end{proof}

\subsection*{I. Sufficient parameter domain for Theorem~\ref{thm5}}
For completeness, we summarize a sufficient set of parameter conditions under which Lemmas~\ref{lemma:2yinU}--\ref{lemma:neighboringdatasets} and Theorem~\ref{thm5} hold. Throughout the analysis we assume:

\begin{enumerate}
    \item $\gamma > 1$ and $t \le k$ so that
    \[
    1 + \frac{1}{k} < 1 + \frac{\gamma}{t} \le e^{\varepsilon_0}\Bigl(1+\frac{\gamma}{t}\Bigr) = e^\varepsilon,
    \]
    which is used in Lemma~\ref{lemma5:random1};
    \item $k$ is large enough compared with the minimum dataset size (e.g., $k \le |D|$ for all datasets considered) so that the lower bounds on $|D|$ in the proofs are satisfied;
    \item the distribution of $L$ is chosen so that the privacy-test probabilities $\mathrm{pt}(D,j,y)$ obey the multiplicative bounds in Lemma~\ref{lemma4:djpartition} (this requirement is encoded in $\varepsilon_0$).
\end{enumerate}
Let
\begin{equation}
\begin{split}
\mathcal{A} &= \Bigl\{(\varepsilon_0,\gamma,t,k):\ \gamma>1,\ t\le k,\ k \le |D|\\
&\quad\text{for all datasets considered,}\\
&\quad\text{and Lemmas~\ref{lemma4:djpartition}, \ref{lemma5:random1} hold}\Bigr\}.
\end{split}
\end{equation}
In all reported experiments that explicitly use the privacy-test mechanism $\mathcal{F}$, the parameters $(\varepsilon_0,\gamma,t,k)$ are chosen to satisfy the above conditions so that the assumptions of Theorem~\ref{thm5} apply.

\paragraph{Verification that reported experiments use parameters in $\mathcal{A}$.}
The main-text experiment that explicitly invokes the mechanism $\mathcal{F}$ and its privacy test is the \textbf{privacy test pass rate} experiment in Figure~\ref{fig:generalpic}. In that experiment we set $t=2$, choose $k \in \{10,20,30,50\}$, and vary $\gamma$ over a range with $\gamma>1$. The dataset size is chosen large enough that $k \le |D|$ for the reported runs, and $\varepsilon_0>0$ is set so that the assumed tail conditions for $L$ are satisfied. Therefore the reported parameter settings are consistent with the sufficient conditions summarized above. Any additional exploratory runs outside these settings are empirical only and are not used as theorem-level instantiations of Theorem~\ref{thm5}.

\section{SUPPLEMENTARY CONCEPTUAL BACKGROUND}
\label{app:supp_conceptual_bg}

\subsection{Noisy SGD}
\label{app:noisy_sgd}
Noisy SGD is a standard differentially private learning paradigm based on sensitivity control and composition. We mention it only as conceptual background; the present paper focuses on non-interactive release rather than private optimization.

\section{SUPPLEMENTARY IMAGE-NOISE AND ADDITIONAL EXPERIMENTS}
\label{app:image_noise_supp}

\subsection{Image noise utility (PSNR/MAE)}
\label{app:image_noise_utility}
The following image-noise experiments are included as supporting material; the main text focuses on entropy calibration and the synthetic-data mechanism $\mathcal{F}$.

\subsection{Privacy--utility tradeoff (image noise)}
\label{app:image_noise_tradeoff}
To quantify the tradeoff between the privacy parameter $\varepsilon$ and data utility, we add Laplace and Gaussian noise to a synthetic $64 \times 64$ image and measure PSNR (dB) and MAE over a range of $\varepsilon$, following standard practice in differential privacy for images. The left panel shows PSNR and the right panel shows MAE; both are evaluated for $\varepsilon \in [0.1, 4]$ on a $64 \times 64$ synthetic image. The figure illustrates the $(\varepsilon,\delta)$-DP utility--privacy tradeoff: smaller $\varepsilon$ implies more noise (lower PSNR, higher MAE). At small $\varepsilon$ (e.g., 0.1) the image is heavily perturbed; at $\varepsilon = 4$ the curves approach a plateau. Laplace and Gaussian behave similarly; small differences are due to the different noise distributions.

\begin{figure}[htbp]
	\centering
	\includegraphics[width=\columnwidth]{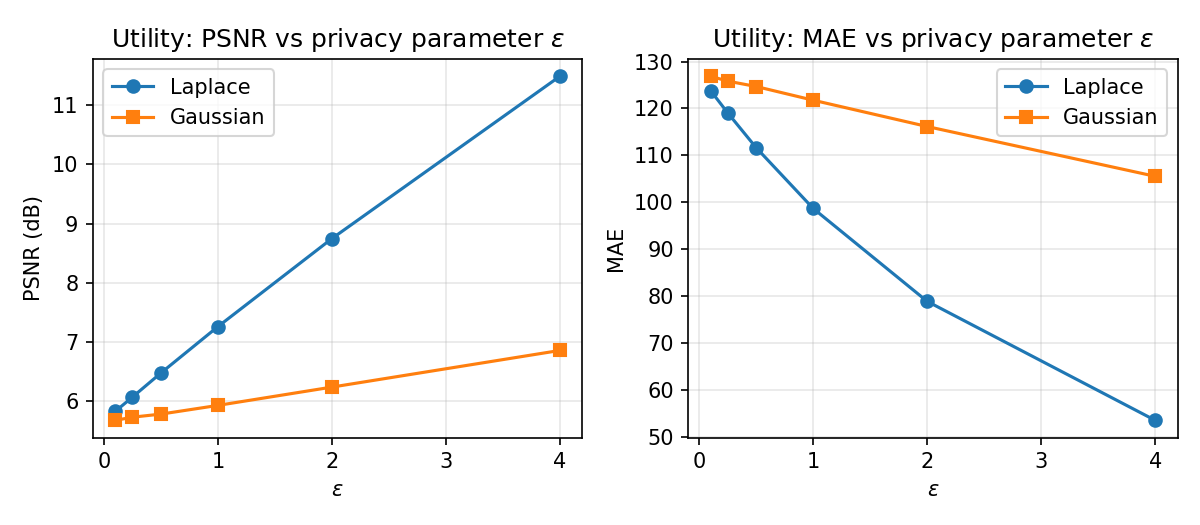}
	\caption{PSNR and MAE vs.\ $\varepsilon$ (Laplace and Gaussian).}\label{fig:epsilon_utility}
\end{figure}

\subsection{Metrics comparison: Laplace vs.\ Gaussian}
\label{app:image_noise_metrics}
\begin{figure*}[htbp]
	\centering
	\includegraphics[width=\textwidth]{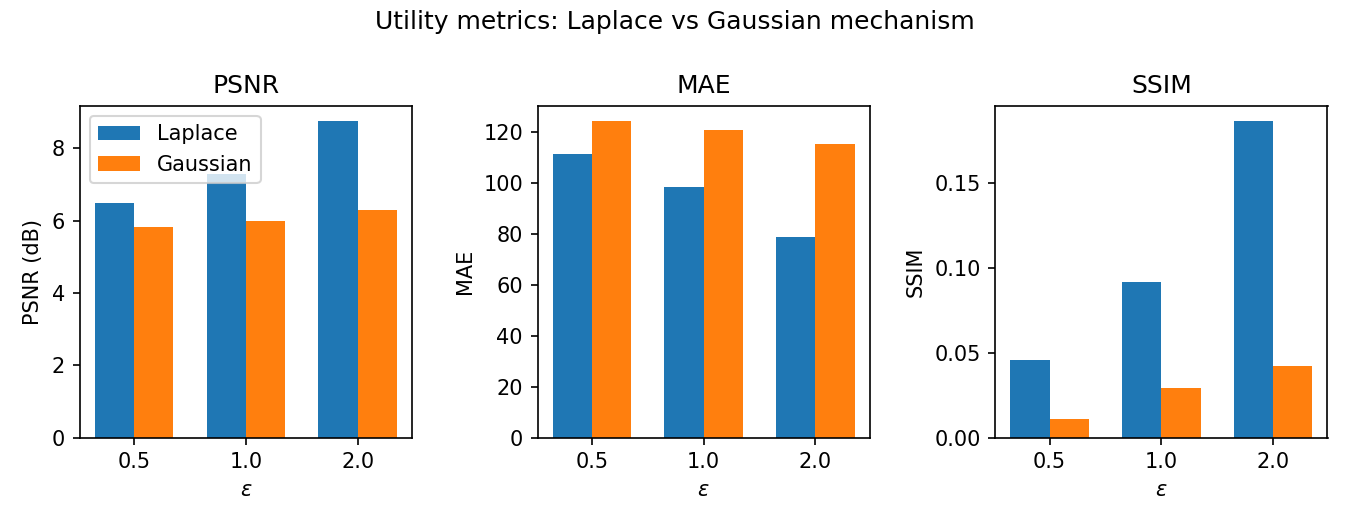}
	\caption{PSNR, MAE, and SSIM at $\varepsilon = 0.5, 1.0, 2.0$ (Laplace vs.\ Gaussian).}\label{fig:metrics_comparison}
\end{figure*}
At each $\varepsilon$, both mechanisms achieve comparable PSNR, MAE, and SSIM. As $\varepsilon$ increases from 0.5 to 2.0, all three metrics improve (PSNR and SSIM increase, while MAE decreases). The SSIM-like metric remains in a plausible range and correlates with PSNR and MAE. This multi-metric comparison supports the use of either Laplace or Gaussian for image release, with the choice depending on whether pure $\varepsilon$-DP or $(\varepsilon,\delta)$-DP is required and on implementation considerations.

\section{Figure and table data descriptions}\label{app:figdata}
Table~\ref{tab:figures_theory} summarizes, for each figure, the data shown and its experimental role in the paper.
\begin{table*}[htbp]
\centering
\caption{Supplementary figures and their experimental role (data shown and connection to the theory).}
\label{tab:figures_theory}
\begin{tabular}{clp{5.2cm}}
\hline
\textbf{Figure} & \textbf{Data} & \textbf{Experimental role} \\
\hline
\ref{fig:generalpic} & Pass rate vs.\ $\gamma$; curves for $k \in \{10,20,30,50\}$ & Theorem~\ref{thm5} (mechanism $\mathcal{F}$ is $(\varepsilon,\delta)$-DP); privacy test. \\
\ref{fig:GaussianKernel1} & Original vs.\ private RKHS mean; $\rho = 10^{-6}, 0.001, 0.1$ & Wiener kernel mechanism; $\rho$ utility--privacy tradeoff. \\
\ref{fig:epsilon_utility} & PSNR, MAE vs.\ $\varepsilon$; Laplace \& Gaussian & $(\varepsilon,\delta)$-DP utility--privacy tradeoff. \\
\ref{fig:entropy_bound} & $\Delta_H$ bound vs.\ $n$ & Theorem~\ref{thm4} (entropy sensitivity). \\
\ref{fig:metrics_comparison} & PSNR, MAE, SSIM at $\varepsilon = 0.5, 1, 2$ & Laplace vs.\ Gaussian. \\
\ref{fig:csv_entropy} & Histogram \& entropy (Table~\ref{tab:csv_entropy}) & Theorem~\ref{thm4}; Laplace histogram release. \\
\ref{fig:csv_entropy_multi} & Multi-dataset entropy (Table~\ref{tab:csv_entropy_multi}) & Generality of entropy--DP framework across domains. \\
\ref{fig:baseline} & Entropy error \& MAE vs.\ $\varepsilon$; Laplace, Gaussian, DP synthetic (Laplace/Gaussian); $\Delta_H$ bound & Baseline comparison; supports Theorem~\ref{thm4}. \\
\ref{fig:ablation} & Entropy error vs.\ bins $m$; $\Delta_H$ bound & Ablation on bins and parameters. \\
\ref{fig:mia} & MIA accuracy (95\% CI), AUC, linkage-style accuracy vs.\ $\varepsilon$ & MIA and linkage attack evaluation; membership-inference robustness. \\
\ref{fig:delta_h_empirical} & $\widehat{\Delta H}$ mean/max vs.\ $\Delta_H$ bound; ratio by $n$ & Theorem~\ref{thm4} alignment: empirical sensitivity below bound. \\
\ref{fig:fig8} & (See caption) & (See caption) \\
\hline
\end{tabular}
\end{table*}

\section{Additional experimental result}\label{app:extrafig}
\begin{figure}[htbp]
	\centering
	\includegraphics[width=\columnwidth]{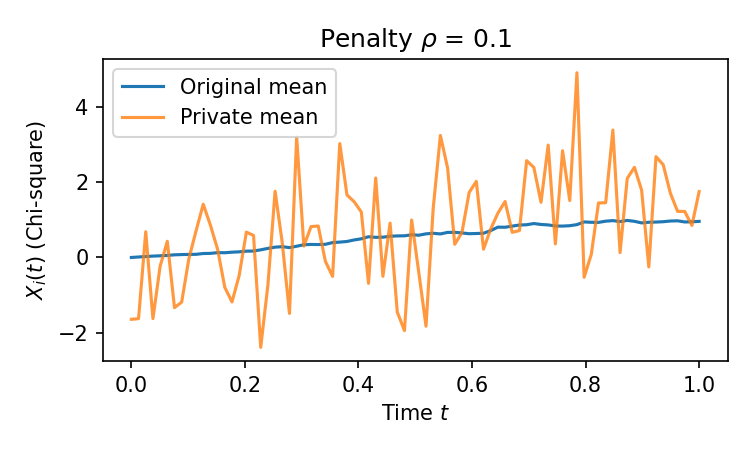}
	\caption{Additional experiment (see text).}\label{fig:fig8}
\end{figure}
Figure~\ref{fig:fig8} provides an additional experimental result (e.g., extended comparison, parameter sweeps, or ablation) and is included for completeness.

\section{Full multi-dataset validation results}\label{app:multidata}
\begin{table*}[htbp]
\centering
\caption{Multi-dataset real-data validation: same entropy+DP pipeline on seven datasets ($\varepsilon=1$, 30 bins; $n$ capped at 100k where applicable).}
\label{tab:csv_entropy_multi}
\begin{tabular}{lrrrrrr}
\hline
Dataset (column) & $n$ & bins & $\varepsilon$ & $H_{\mathrm{orig}}$ & $H_{\mathrm{noisy}}$ & $\Delta_H$ bound \\
\hline
y\_amazon-google-large ($y$) & 100000 & 30 & 1.0 & 3.0279 & 3.0286 & 0.00037 \\
house-prices-train (SalePrice) & 1460 & 30 & 1.0 & 3.4873 & 3.5155 & 0.01676 \\
home-credit-application-train (AMT\_INCOME\_TOTAL) & 100000 & 30 & 1.0 & 0.0004 & 0.0034 & 0.00037 \\
tabular-feature-engineering ($y_1$) & 10000 & 30 & 1.0 & 4.6645 & 4.6640 & 0.00300 \\
r\_diff\_train ($y_1$) & 10000 & 30 & 1.0 & 0.0237 & 0.0459 & 0.00300 \\
pow\_train ($y_1$) & 10000 & 30 & 1.0 & 4.7066 & 4.7065 & 0.00300 \\
unemployment\_data (UNRATE) & 601 & 30 & 1.0 & 4.0457 & 4.0894 & 0.03645 \\
\hline
\end{tabular}
\end{table*}

\begin{figure*}[htbp]
	\centering
	\subfigure[Amazon-Google reviews ($y$).]{\includegraphics[width=0.48\textwidth]{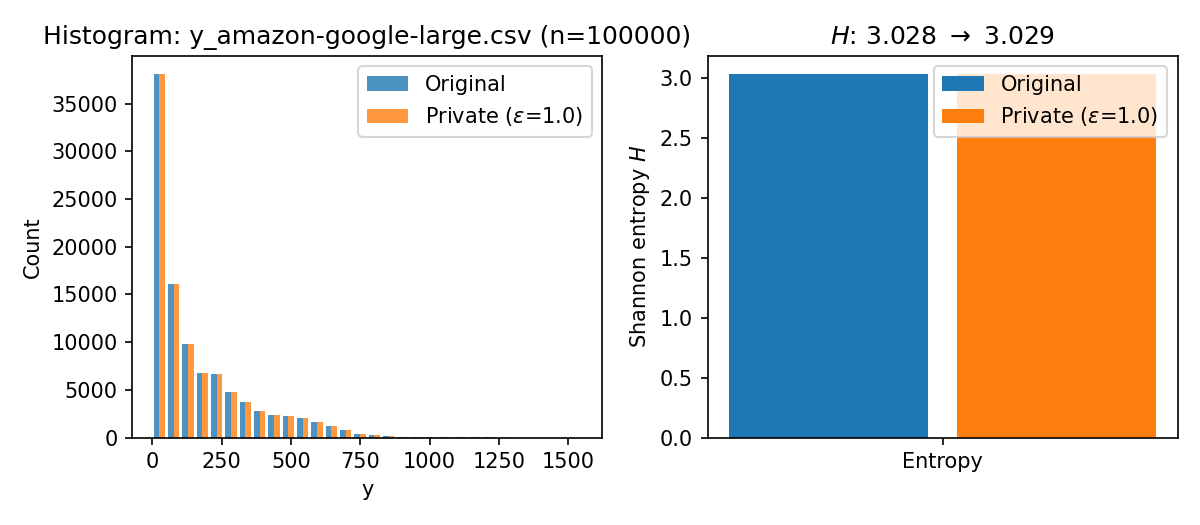}}
	\hfill
	\subfigure[House prices (SalePrice).]{\includegraphics[width=0.48\textwidth]{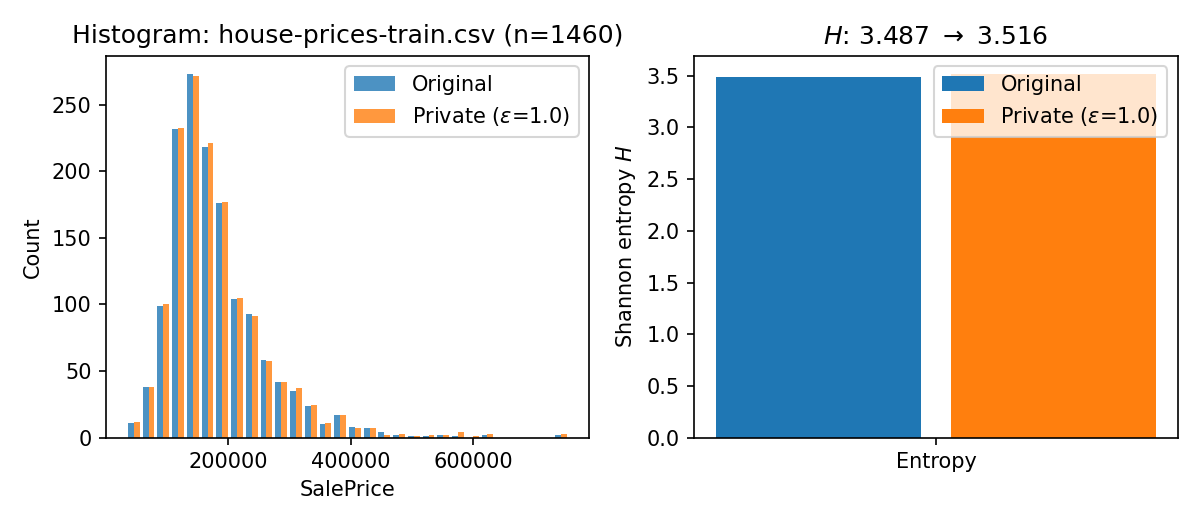}}\\
	\subfigure[Tabular feature engineering ($y_1$).]{\includegraphics[width=0.48\textwidth]{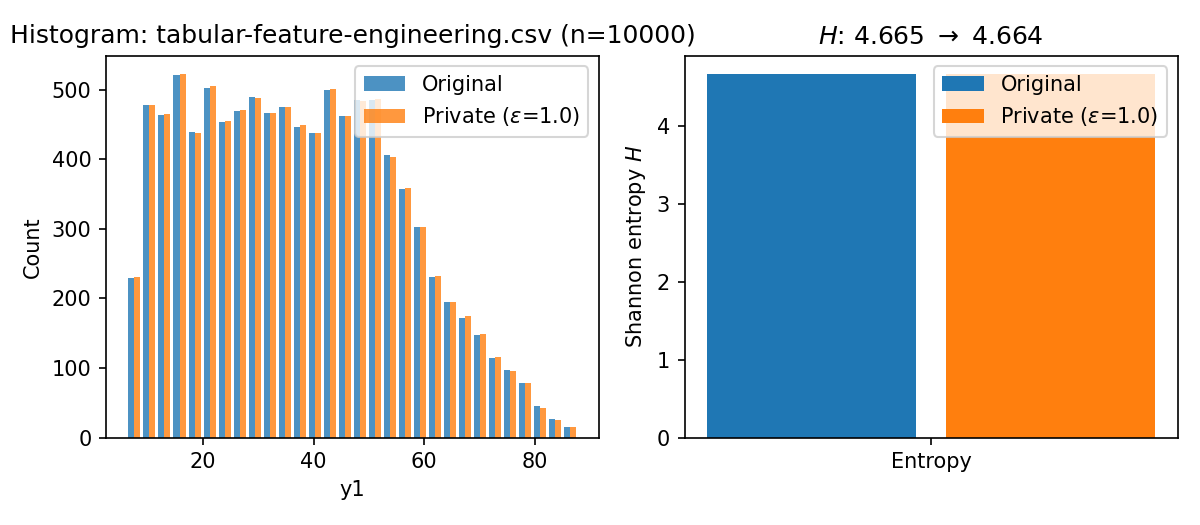}}
	\hfill
	\subfigure[U.S. unemployment (UNRATE).]{\includegraphics[width=0.48\textwidth]{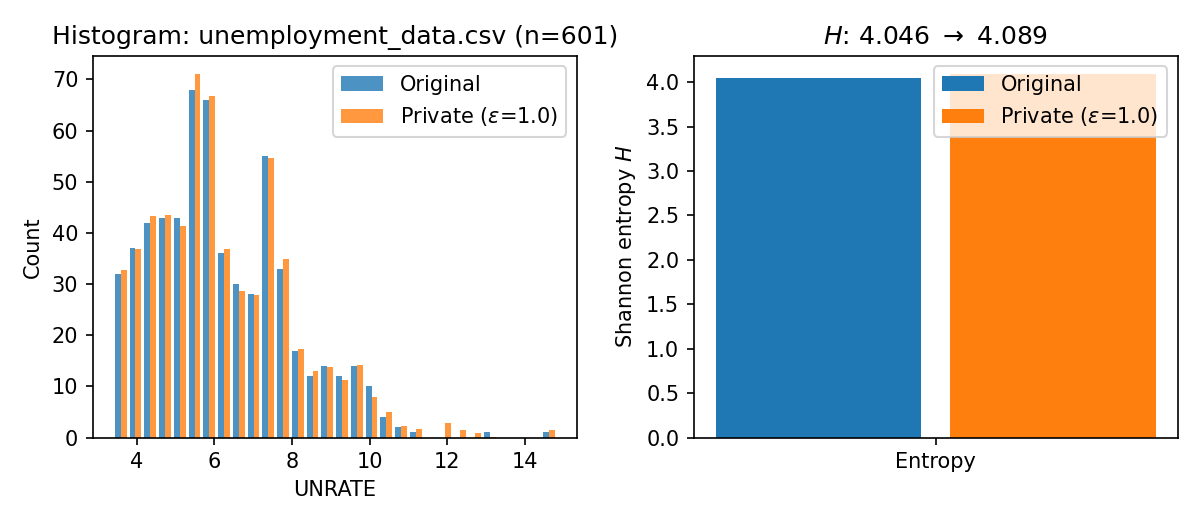}}
	\caption{Multi-dataset validation: histogram and entropy (original vs.\ $\varepsilon$-DP) on four representative public datasets. Same mechanism and Theorem~\ref{thm4} bound; full table in Table~\ref{tab:csv_entropy_multi}.}\label{fig:csv_entropy_multi}
\end{figure*}

\section*{Code availability}
The code implementing the experiments and algorithms described in this paper is publicly available~\cite{reaedp2025code}.